\documentclass[11pt]{article}

\usepackage[margin=1in]{geometry}
\usepackage{amsmath,amssymb,amsthm,mathtools}
\usepackage{enumitem}
\usepackage[colorlinks=true,linkcolor=blue,citecolor=blue,urlcolor=blue]{hyperref}
\usepackage[nameinlink,capitalize]{cleveref}
\usepackage{microtype}
\usepackage{thmtools}
\usepackage{thm-restate}
\usepackage{xcolor}
\usepackage{xurl}
\usepackage{graphicx}
\usepackage{subcaption}
\usepackage{mdframed}
\newtheorem{theorem}{Theorem}[section]


\newtheorem{lemma}[theorem]{Lemma}
\newtheorem{claim}[theorem]{Claim}

\newtheorem{definition}[theorem]{Definition}
\newtheorem{proposition}[theorem]{Proposition}
\newtheorem{fact}[theorem]{Fact}

\newcommand{\calF}{\mathcal F}
\newcommand{\calL}{\mathcal L}
\newcommand{\calO}{\mathcal O}
\newcommand{\calQ}{\mathcal Q}
\newcommand{\frakp}{\mathfrak p}
\newcommand{\fraka}{\mathfrak a}
\newcommand{\frakd}{\mathfrak d}
\newcommand{\frakA}{\mathfrak A}

\newcommand{\frakP}{\mathfrak P}
\newcommand{\id}{\mathrm{id}}
\newcommand{\ur}{\mathrm{ur}}

\newcommand{\class}{\mathrm{class}}

\DeclareMathOperator{\Tr}{\operatorname{Tr}}

\DeclareMathOperator{\disc}{\operatorname{disc}}
\newcommand{\val}{\operatorname{val}}
\newcommand{\poly}{\operatorname{poly}}

\newcommand{\Gal}{\operatorname{Gal}}
\newcommand{\Aut}{\operatorname{Aut}}
\newcommand{\CRR}{\operatorname{CRR}}
\newcommand{\Norm}{\operatorname{N}}
\newcommand{\rd}{\operatorname{rd}}
\newcommand{\classnum}{\operatorname{h}}

\newcommand{\N}{\mathbb N}
\newcommand{\Z}{\mathbb Z}
\newcommand{\Q}{\mathbb Q}
\newcommand{\R}{\mathbb R}
\newcommand{\C}{\mathbb C}

\renewcommand{\Re}{\operatorname{Re}}
\renewcommand{\Im}{\operatorname{Im}}
\newcommand{\F}{\mathbb F}
\newcommand{\OO}{\mathcal O}

\newcommand{\eps}{\varepsilon}

\newcommand{\ceil}[1]{\left\lceil #1\right\rceil}
\newcommand{\floor}[1]{\left\lfloor #1\right\rfloor}
\newcommand{\set}[1]{\{ #1 \}}

\crefname{theorem}{Theorem}{Theorems}
\crefname{lemma}{Lemma}{Lemmas}
\crefname{proposition}{Proposition}{Propositions}
\crefname{corollary}{Corollary}{Corollaries}
\crefname{definition}{Definition}{Definitions}
\crefname{remark}{Remark}{Remarks}
\Crefname{theorem}{Theorem}{Theorems}
\Crefname{lemma}{Lemma}{Lemmas}
\Crefname{proposition}{Proposition}{Propositions}
\Crefname{corollary}{Corollary}{Corollaries}
\Crefname{definition}{Definition}{Definitions}
\Crefname{remark}{Remark}{Remarks}

\title{Furthest Pair Requires Quadratic Time in Superconstant Dimension under SETH}
\author{
Barna Saha\thanks{Supported by NSF HDR TRIPODS Phase II grant 2217058 (EnCORE Institute).}  \\
University of California San Diego  \\
\texttt{bsaha@ucsd.edu}
\and
Yinzhan Xu\footnotemark[1]  \\
University of California San Diego  \\
\texttt{xyzhan@ucsd.edu}
\and
Christopher Ye\footnotemark[1]  \\
University of California San Diego  \\
\texttt{czye@ucsd.edu}
}
\date{}

\begin{document}
\maketitle

\begin{abstract}
Several fundamental problems in computational geometry admit algorithms with running time $f(d)\cdot n^{2-\Theta(1/d)}$ for $n$ points in $d$ dimensions, making them among the most prominent examples of barely subquadratic computation. Notable members of this class include Furthest Pair, Bichromatic Closest Pair, (Bichromatic) Maximum Inner Product, and Hopcroft’s Problem. Chen [Theory Comput. 2020] proved that, assuming the Strong Exponential Time Hypothesis (SETH), these problems require $n^{2-o(1)}$ time when the dimension satisfies $d=2^{\Theta(\log^* n)}$. We extend this lower bound to all efficiently constructible dimensions $d=\omega(1)$. Thus, assuming SETH, the dependence of the best known algorithms on the dimension is essentially unavoidable. The proof utilizes techniques in OpenAI's recent disproof of the Erd\H{o}s unit distance conjecture. 

The proof was initially discovered by ChatGPT 5.5 Pro. The authors have validated and substantially edited the proof to improve the presentation. 
\end{abstract}

\section{Introduction}

Furthest Pair and Closest Pair are among the most fundamental problems in computational geometry. Given a set of $n$ points in $d$ dimensions, Furthest Pair (resp. Closest Pair) asks to find a pair of points that are furthest (resp. closest) to each other in Euclidean distance. Their applications reach far beyond computational geometry \cite{DBLP:conf/focs/ShamosH75, toussaint1983solving, DBLP:journals/ijcga/MalandainB02, DBLP:conf/compgeom/Har-Peled01}, ranging over spatial databases \cite{DBLP:conf/sigmod/CorralMTV00, DBLP:journals/dke/CorralMTV04, DBLP:conf/adbis/Garcia-GarciaCI16, DBLP:conf/pci/MavrommatisMC21}, data mining \cite{DBLP:journals/jea/Eppstein00, DBLP:conf/vldb/ManolopoulosT01, gonzalez1985clustering, xiao2011efficient}, computer vision \cite{DBLP:journals/vc/ShapiraSC08, chen2018fast}, 
and robotics \cite{lin1991fast, DBLP:journals/trob/GilbertJK88, DBLP:journals/jcise/RedonLMK05}.

Despite analogous definitions, Closest Pair and Furthest Pair have  markedly different algorithmic complexities, observed as early as 1976 \cite{DBLP:conf/stoc/BentleyS76}: Currently, the best algorithms for Closest Pair run in $f(d) \cdot n \log^{O(d)}(n)$ time \cite{DBLP:conf/stoc/BentleyS76,DBLP:journals/iandc/KhullerM95,DBLP:journals/jal/DietzfelbingerHKP97} (which is truly subquadratic as long as $d = o(\log n / \log \log n)$), while the best algorithms for Furthest Pair run in $f(d) \cdot n^{2-\Theta(1/d)}$ time \cite{yao1982constructing,DBLP:conf/compgeom/AgarwalESW90} (which stops being truly subquadratic for any superconstant $d$). 

Furthest Pair is not the only fundamental problem with a $f(d) \cdot n^{2-\Theta(1/d)}$ running time. Other examples include:
\begin{itemize}
    \item \textit{Bichromatic Closest Pair} \cite{DBLP:conf/compgeom/AgarwalESW90}. In this problem, we are given two sets of points $A,B \subseteq \Z^d$, and we need to find $a \in A, b \in B$ that minimizes $\lVert a - b\rVert_2$. A priori, it may be surprising that the complexity of Bichromatic Closest Pair is more comparable to that of Furthest Pair than to that of Closest Pair. Similar to Furthest Pair and Closest Pair, Bichromatic Closest Pair is a fundamental problem with applications in various domains. In fact, some of the applications listed earlier specifically apply to this Bichromatic version of  Closest Pair (e.g., \cite{lin1991fast, DBLP:journals/trob/GilbertJK88}).
    \item \textit{Maximum Inner Product} ($\Z$-Max-IP) \cite{DBLP:conf/compgeom/AgarwalESW90,yao1982constructing,DBLP:journals/dcg/Matousek92}. In this problem, we are given two sets of points $A,B \subseteq \Z^d$, and we need to find $a \in A, b \in B$ that maximizes $a \cdot b$. It is one of the central problems in similarity search with a myriad of  applications \cite{DBLP:journals/toc/Har-PeledIM12, DBLP:journals/cacm/AndoniI08, DBLP:conf/nips/RahimiR07, DBLP:conf/kdd/RamG12, DBLP:conf/nips/Shrivastava014, DBLP:conf/soda/AndoniINR14, DBLP:conf/nips/AndoniILRS15, DBLP:conf/stoc/AndoniR15, DBLP:conf/icml/NeyshaburS15, DBLP:conf/www/Shrivastava015, DBLP:journals/jacm/Valiant15, DBLP:conf/focs/AlmanW15, DBLP:journals/talg/KarppaKK18, DBLP:conf/pods/AhlePR016, DBLP:journals/tods/TeflioudiG17, DBLP:conf/stoc/ChristianiP17, DBLP:conf/soda/Christiani17}. 
    \item \textit{Hopcroft’s problem  (also known as $\Z$-OV)} \cite{DBLP:journals/dcg/Matousek93,DBLP:journals/dcg/Chazelle93a}. In this problem, we are given two sets of points $A,B \subseteq \Z^d$, and we need to determine whether there exists $a \in A, b \in B$ such that $a \cdot b = 0$. Hopcroft’s problem, first posed by John Hopcroft in the early 1980s, is a basic problem related to a range of computational geometry problems \cite{DBLP:conf/cccg/Erickson95}. Finding faster algorithms for Hopcroft’s problem in various settings remains an active research direction \cite{DBLP:journals/talg/ChanZ24, DBLP:conf/mfcs/AndrejevsBV24}. 
\end{itemize}

Prior work \cite{DBLP:conf/soda/Williams18, DBLP:journals/toc/Chen20} has studied the fine-grained hardness of these problems under the following Orthogonal Vectors Hypothesis (OVH), which is known to hold under the Strong Exponential Time Hypothesis (SETH) \cite{DBLP:journals/tcs/Williams05}. 

\begin{definition}[OVH]
For every $\delta>0$ there is a constant $c=c(\delta)$ such that any algorithm that, given two sets of $n$ Boolean vectors in dimension $c\log n$, determines whether there is a pair of orthogonal vectors cannot run in time $O(n^{2-\delta})$.
\end{definition}

Straightforward transformations \cite{DBLP:journals/tcs/Williams05,DBLP:conf/focs/AlmanW15} from Orthogonal Vectors show that Furthest Pair, Bichromatic Closest Pair, Maximum Inner Product, and Hopcroft’s problem require $n^{2-o(1)}$ time when $d = \omega(\log n)$ under OVH. Williams \cite{DBLP:conf/soda/Williams18} improved the dependence to $d = \poly \log \log n$, and Chen \cite{DBLP:journals/toc/Chen20} further improved it to $d = 2^{\Theta(\log^* n)}$. However, there remained a key gap between the upper and lower bounds: 
\begin{center}
    \textit{
    Is there a superconstant $d$ where any of Furthest Pair, Bichromatic Closest Pair, Maximum Inner Product, or Hopcroft’s problem in dimension $d$ can be solved in truly subquadratic time? 
    }
\end{center}

This paper resolves this gap and shows that these problems are already hard for any superconstant dimension (as long as the dimension is efficiently computable). 

\begin{theorem}[Informal \cref{thm:main}]
\label{thm:main-intro}
    Under OVH or SETH, Furthest Pair, Bichromatic Closest Pair, Maximum Inner Product, and Hopcroft’s problem  (also known as $\Z$-OV) require $n^{2-o(1)}$ time for any efficiently constructible dimension $d = \omega(1)$. 
\end{theorem}

Some of these problems, in particular Maximum Inner Product and Bichromatic Closest Pair, have been reduced to other problems in a black-box way, including some linear algebra problems on geometric graphs \cite{DBLP:conf/focs/AlmanCS020},  dynamic detection of firing neurons~\cite{DBLP:conf/nips/AlmanL00Z23}, and attention computation \cite{gupta2026subquadratic}. Hence, 
\cref{thm:main-intro} also directly implies improved lower bounds for these downstream applications.

\subsection{Technical Overview}

We now give an overview of the techniques in this work.
The hardness of exact $\Z$-Max-IP, Furthest Pair and Bichromatic Closest Pair follows from reductions from $\Z$-OV, based on reductions in \cite{DBLP:conf/soda/Williams18, DBLP:journals/toc/Chen20}. See \cref{app:reductions} for details. Hence, we focus on $\Z$-OV. 

Recall in the $\Z$-OV problem, we are given two sets of points $A,B \subseteq \Z^d$, and we need to determine whether there exists $x \in A, y \in B$ such that $x \cdot y = 0$.
The key technical contribution of this work is to design a better reduction from OV to $\Z$-OV.

Generally, such reductions proceed as follows: given Boolean vectors $x, y$, split each $d$-dimensional vector into $L$ blocks of size $b$.
Then map each $b$-bit block $x^{(i)}$ in $x$ and each $b$-bit block $y^{(i)}$ in $y$ into $w = O(1)$ integers using (efficiently computable) maps $\alpha$ and $\beta$ respectively. 
For a $d = bL$-dimensional vector $x$, $\alpha(x)$ applies the map $\alpha$ to each block.
This way, $\alpha(x) \cdot \beta(y) = \sum_{i = 1}^{L} \alpha(x^{(i)}) \cdot \beta(y^{(i)})$. 
The goal is to show the existence of a set $V$ such that
\begin{equation*}
    \alpha(x) \cdot \beta(y) \in V \quad\Longleftrightarrow\quad x \cdot y = 0 \text{.}
\end{equation*}
Checking membership in $V$ can be simulated by $|V|$ calls to $\Z$-OV (with one extra dimension), since for a fixed $v \in V$, $x \cdot y = v$ iff $(x, v) \cdot (y, -1) = 0$.
Thus, the reduction is efficient as long as $|V| = n^{o(1)}$ and the entries of the mapped vectors are small.

Consider the first step in \cite{DBLP:journals/toc/Chen20}'s reduction.
Given $b$ distinct primes $p_1, \dots, p_b > L$, the Chinese Remainder Theorem (CRT) gives a one-to-one correspondence $\alpha$ between $\set{0, 1}^{b}$ and a subset of $[0, \prod p_{j})$ (note that in the approach, $\beta=\alpha$).
In particular, for the $i$-th block of $x$, $\alpha(x^{(i)}) \equiv x_{j}^{(i)} \pmod{p_{j}}$ for $j \in [b]$, i.e. we can recover the $j$-th bit of the $i$-th block of $x$ from $\alpha(x^{(i)})$ modulo $p_j$.
In other words, 
\begin{equation*}
    \alpha(x) \cdot \alpha(y) = \sum_{i = 1}^{L} \alpha(x^{(i)}) \cdot \alpha(y^{(i)}) \equiv \sum_{i = 1}^{L} x^{(i)}_{j} y^{(i)}_{j} \pmod{p_{j}} \text{.}
\end{equation*}
Since the latter sum is in $[0, L] \subseteq [0, p_{j})$, it is zero iff $\sum_{i = 1}^{L} x^{(i)}_{j} y^{(i)}_{j} = 0$.
In other words, $x \cdot y = 0$ iff $\alpha(x) \cdot \alpha(y) \equiv 0 \pmod{p_j}$ for all $j \in [b]$.
Furthermore, the values of the product $\alpha(x) \cdot \alpha(y)$ are at most $L \cdot \left(\prod p_{j}\right)^2$ so it suffices to take $V$ to be the set of integers in this range that are $0$ modulo  all $p_{j}$.
Here, we note that typically (e.g. \cite{DBLP:journals/toc/Chen20}) $|V|$ is simply bounded by the range of $\alpha(x) \cdot \alpha(y)$.
Unfortunately, since we need $b$ distinct primes, we have $\prod p_j = b^{\Omega(b)}$.
Thus, if we want to use the same bound $|V| \leq L (\prod p_j)$ we need $b^{b} = n^{o(1)}$ to ensure an efficient reduction.
Since $bL = d = \Theta(\log n)$, the requirement $b^{b} = n^{o(1)}$ forces $b \log b = o(\log n) = o(b L)$ or $L = \omega(\log b)$.
The above approach establishes hardness of $\Z$-OV in dimension $\omega(\log \log n)$.
Through a recursive refinement of this approach, \cite{DBLP:journals/toc/Chen20} achieves hardness for dimensions $2^{\Theta(\log^* n)}$.

At a high level, the previous approach was bottlenecked by the density of primes in $\Z$.
In the reduction above, we use the fact that there are $b$ distinct primes bounded by $O(b \log b)$.
If, hypothetically, one had access to $b$ distinct primes bounded by a constant, then one could use CRT to map $b$ bits into integers of size $2^{O(b)}$, which is $2^{o(bL)} = n^{o(1)}$ for any superconstant $L$.
Interestingly, a similar density barrier appears in Erd\H{o}s's $n^{1+c / \log \log n}$ lower bound  of the unit distance problem~\cite{erdos1946sets}, where the relevant primes are those congruent to $1$ modulo $4$. 
The breakthrough result of \cite{UnitDistance} on the unit distance problem bypasses this barrier, and instead constructs number fields in which selected rational primes split into many prime ideals.
Indeed, the number theory machinery of \cite{UnitDistance} is exactly what allowed us to overcome the barrier faced by previous attempts to establish the hardness of $\Z$-OV. 

We now briefly discuss how the number theory machinery helps us achieve the desired reduction.
First, a few definitions (we will try to keep these as minimal as possible).
A \emph{ring} is a set closed under addition and multiplication (the classical example is the integers $\Z$).
An \emph{ideal} is a set closed under addition and scalar multiplication.
An ideal $\frakp$ is prime if whenever $ab \in \frakp$, either $a \in \frakp$ or $b \in \frakp$.
A \emph{field} is a ring with multiplicative inverses (the classical example is $\Q$).
A (finite) \emph{field extension} is obtained by adjoining to $\Q$ the roots of some polynomials (for example, $\Q(\sqrt{2})$ is the extension obtained by adjoining roots of $x^2 - 2$).
Every number field is equipped with a \emph{ring of integers}, the set of elements that are roots of monic polynomials with integer coefficients.
If we take a prime $q \in \Z$ and the ideal generated by $q$ in $\calO_{K}$, $q$ may no longer be prime, but it factors uniquely into prime ideals (analogous to how integers have a unique prime factorization in $\Z$).
A prime $q$ \emph{splits completely} in $K$ if the factorization of $q$ in $\calO_{K}$ has (1) no repeated factors, and (2) for every factor $\frakp$, the residue field $\calO_{K}/\frakp \cong \F_{q}$.
An \emph{embedding} of a field extension $K$ is a homomorphism $K \hookrightarrow \C$ that fixes $\Q$.

\Cref{prop:codebook} shows that for any sufficiently large $b$ and $L$, there is a degree $O(b)$ field extension $K$, its corresponding ring of integers $\calO_{K}$ and a \emph{constant} number of primes $\calQ = \set{q_1, \dots, q_T}$, each of size $\Theta(\sqrt{L})$ satisfying:
\begin{enumerate}
    \item Each $q_t$ splits completely into $\Theta(b)$ distinct prime ideals in $\calO_{K}$.
    \item For each $j \in [b]$, there are two primes $q_{j, 1}, q_{j, 2} \in \calQ$ and corresponding prime ideal factors $\frakp_{j, 1}, \frakp_{j, 2}$. 
    \item For every bit vector $a \in \set{0, 1}^{b}$, there is a corresponding element $u_a \in K \setminus \set{0}$ such that:
    \begin{enumerate}
        \item $|\sigma(u_a)| = 1$ for every embedding $\sigma$.
        \item $Q^2 u_a \in \calO_{K}$ where $Q := \prod_{t = 1}^{T} q_t$.
        \item For every prime ideal $\frakp_{j, r}$, $\val_{\frakp_{j, r}}(u_a) = 2 a_j$ and $\val_{c\frakp_{j, r}}(u_a) = \val_{\frakp_{j, r}}(c(u_a)) = - 2 a_j$, where $\val_{\frakp}(x)$ denotes the valuation of $x$ at $\frakp$ (roughly speaking, how many times $\frakp$ divides $x$)\footnote{We note that a valuation can be negative. For example, in the field $\Q$, $1/9 = 3^{-2}$.} and $c$ denotes the complex conjugation automorphism on $K$ (e.g. $i \mapsto -i$ if $K = \Q(i)$).
    \end{enumerate}
\end{enumerate}

\paragraph{Embedding $\set{0, 1}^{b}$ into $\calO_{K}$.}
Our first step is to embed binary $b$-bit vectors into the ring of integers $\calO_{K}$.
Formally, for each prime $q_t$, map $a \in \set{0, 1}^{b}$ to $F_{t}(a) := (Q^2 c(u_a))^{q_t - 1} \in \calO_{K}$.\footnote{In the actual proof, we define $F_t$ as a map on $u_a$.}
Define $G_t(a) := Q^2 c(u_a)$ so that $F_t(a) = G_t(a)^{q_t - 1}$.
Define $K_t(x, y) = \sum_{i = 1}^{L} F_t(x^{(i)}) F_t(y^{(i)})$. 
If $x \cdot y = 0$ and $u \cdot v > 0$, then 
we claim $K_t(x, y) \neq K_t(u, v)$ for some $t$.
Note that if we apply the map $F_t$ to each block, the dot product of the resulting vector is
\begin{equation*}
    \sum_{i = 1}^{L} \sum_{t = 1}^{T} F_t(x^{(i)}) F_t(y^{(i)}) = \sum_{t = 1}^{T} K_t(x, y) \text{.}
\end{equation*}
By additivity of valuations\footnote{In general, $\val_{\frakp}(xy) = \val_{\frakp}(x) + \val_{\frakp}(y)$ and $\val_{\frakp}(x + y) \geq \min(\val_{\frakp}(x), \val_{\frakp}(y))$.} and 3(c) above,
\begin{equation*}
    \val_{\frakp_{j, r}}(G_t(x^{(i)}) G_t(y^{(i)})) = \val_{\frakp_{j, r}}(Q^4 c(u_{x^{(i)}}) c(u_{y^{(i)}})) = 2 (2 - x^{(i)}_j - y^{(i)}_j)
\end{equation*} 
where $x^{(i)}_j$ is the $j$-th bit of $x^{(i)}$, i.e. the valuation is $0$ iff $x^{(i)}_j y^{(i)}_j = 1$ (see \Cref{lem:valuation-table}).
In other words, if $x, y$ are orthogonal, then every prime ideal $\frakp_{j, r}$ divides $G_t(x^{(i)}) G_t(y^{(i)})$ for all $i \in [L]$, while if $x, y$ are not orthogonal, say $x^{(i)}_{j} y^{(i)}_{j} = 1$, then some prime ideal $\frakp_{j, r}$ does not divide $G_t(x^{(i)}) G_t(y^{(i)})$.
In the residue field $\calO_{K}/\frakp_{j, r} \cong \F_{q_t}$,\footnote{Here we use the fact that $q_t$ splits completely.} we have $G_t(x^{(i)}) G_t(y^{(i)}) \equiv 0 \pmod{\frakp_{j, r}}$ for all $i, j$ in the former case and non-zero for some $i, j$ in the latter case.

We now examine $K_t$.
When $x, y$ are orthogonal, $F_t(x^{(i)}) F_{t}(y^{(i)}) \equiv 0 \pmod{\frakp_{j, r}}$ for all selected prime ideals so that $K_t(x, y) \equiv 0 \pmod{\frakp_{j, r}}$.
When $x, y$ are not orthogonal, Fermat's Little Theorem implies that $F_t(x^{(i)}) F_t(y^{(i)}) \equiv 1 \pmod{\frakp_{j, r}}$ for some $i, j$, since any non-zero element taken to the $(q_t - 1)$-th power in $\F_{q_t}$ is $1$.
For such a $j \in [b]$, let
\begin{equation*}
    0 < k := |\set{i \mid x_{j}^{(i)} = y_{j}^{(i)} = 1}| \leq L
\end{equation*}
denote the number of blocks where the $j$-th bit is not orthogonal.
Since $q_{j, 1} q_{j, 2} > L \geq k$ at least one of $q_{j, 1}, q_{j, 2}$ will not divide $k$.
In particular, since each non-zero coordinate contributes $1$ and the remaining terms contribute $0$, we have $K_t(x, y) \equiv k \not\equiv 0 \pmod{\frakp_{j, r}}$ for the corresponding prime $\frakp_{j, r}$ lying above $q_{j, r} \nmid k$.

\paragraph{Embedding $\calO_{K}$ into $\C$.}
Next, we embed $\calO_{K}$ into $\C$ via $\sigma$, with the guarantee that $K_t(x, y) \neq K_t(u, v)$ implies that $|\sigma(K_t(x, y)) - \sigma(K_t(u, v))| \ge \exp(-O(b \sqrt{L} \log L))$, i.e. distinct elements in the ring of integers are mapped to sufficiently separated elements of $\C$ (\Cref{lem:product-separation}).
As before, note that if we apply $\sigma$ entry-wise, the resulting dot product of the mapped vectors is
\begin{equation*}
    \sum_{i = 1}^{L} \sum_{t = 1}^{T} \sigma(F_t(x^{(i)})) \sigma(F_t(y^{(i)})) = \sum_{t = 1}^{T} \sigma(K_t(x, y)) \text{.}
\end{equation*}
Fix $\Delta := K_t(x, y) - K_t(u, v)$.
First, observe that any embedding is not too large:
\begin{equation*}
    |\sigma(F_t(a))| = Q^{2 (q_t - 1)} |\sigma(c(u_a))|^{q_t - 1} = Q^{2(q_t - 1)} = \exp(O(\sqrt{L} \log L)) 
\end{equation*}
since $\sigma(c(u_a)) = \overline{\sigma(u_a)}$ has absolute value $1$ and $Q$ is the product of $O(1)$ primes in the range $\Theta(\sqrt{L})$.
A similar bound holds for $\Delta$.

To conclude, we claim that the product $\prod_{\sigma} |\sigma(\Delta)|$ over all embeddings $\sigma$ is a non-zero integer (i.e. at least $1$).
Note that this immediately ensures $|\sigma(\Delta)| \geq \exp(-O(b \sqrt{L} \log L))$ for any embedding $\sigma$, since the number of embeddings is equal to the degree of $K$, or $O(b)$.
To bound the product, we note that $\Delta \in \calO_{K}$ satisfies a monic polynomial with integer coefficients.
In particular, the product of all its embeddings is (up to sign) a power of the constant coefficient, an integer.

\paragraph{Embedding $\C$ into $\Z$.}
So far, we have mapped each block of $b$ bits into $T$-dimensional vectors in $\C$ via $\sigma(F_t(a))$ with the guarantee that $|\sigma(K_t(x, y)) - \sigma(K_t(u, v))| \geq \delta := \exp(- O(b \sqrt{L} \log L))$ for some $t$ whenever $0 = x \cdot y < u \cdot v$.
For each $q_t$, we construct maps $f_t, g_t: \C \rightarrow \Z^4$ such that $f_t(x) \cdot g_t(y) \simeq M^{2t} \Re(N^2 x y) + M^{2t + 1} \Im(N^2 xy)$ for some sufficiently large $N \gg 1/\delta$ and $M$.
(Here, note that we will not be able to define $f_t, g_t$ to equal the right-hand side exactly, since $xy$ may not be integral, but they will be close to the desired value up to rounding error. We will ignore the rounding errors and treat the above equality as exact in this overview.)
If we apply $f_t, g_t$ entry-wise to the $L$ coordinates corresponding to prime $q_t$, then the dot product of the resulting $\Z$-vectors is
\begin{equation*}
    \sum_{i = 1}^{L} \sum_{t = 1}^{T} f_t(\sigma(F_t(x^{(i)}))) \cdot g_t(\sigma(F_t(y^{(i)}))) = \sum_{t = 1}^{T} M^{2t} \Re(N^2 \sigma(K_t(x, y))) + M^{2t + 1} \Im(N^2 \sigma(K_t(x, y))) \text{.}
\end{equation*}
Now, we observe that $|N^2 \sigma( K_{t}(x, y))| = \poly(1/\delta) \cdot L^{O(\sqrt{L})} = \exp(O(b \sqrt{L} \log L))$ so that if we set $M \gg \exp(O(b \sqrt{L} \log L))$, we can recover the individual summands from the dot product.
By setting $N \gg 1/\delta$, we can see that either the real parts or imaginary parts of $N^2 \sigma(K_t(x, y))$ and $N^2 \sigma(K_t(u, v))$ will disagree if $|\sigma(K_t(x, y)) - \sigma(K_t(u, v))| \geq \delta$.

Thus, combined with our previous guarantees, we can ensure that  (1) the images of vector pairs with $x \cdot y = 0$ and $u \cdot v > 0$ are disjoint, and (2) the entries of the vectors, the size of the possible dot products, and the output set $V$ have size bounded by $\exp(O(b \sqrt{L} \log L))$.
In particular, we can reduce OV to $\exp(O(b \sqrt{L} \log L))$ instances of $O(L)$-dimensional $\Z$-OV.

\paragraph{Statement on AI use.} The initial proof was generated by ChatGPT 5.5 Pro, and we also used other AI systems (Codex, Claude Opus and Gemini) to generate feedback. The initial prompt was essentially the only mathematically meaningful input from the authors, and it is very simple: 
\begin{center}
\begin{minipage}{0.9\linewidth}
\centering
    Try to use this proof idea \url{https://cdn.openai.com/pdf/74c24085-19b0-4534-9c90-465b8e29ad73/unit-distance-proof.pdf} to improve the $2^{O(\log^* n)}$ bound in \url{https://arxiv.org/pdf/1802.02325}. 
\end{minipage}
\end{center}
ChatGPT was not able to solve the problem in its first response, and needed multiple rounds of back and forth. These follow-up rounds included asking the model to continue trying, and asking the model to fix or improve the manuscript based on AI-generated feedback; in later rounds, we used Codex for this process.

The proof uses heavy machinery from algebraic number theory to prove the existence of an embedding from OV to $\Z$-OV (\cref{lem:nonuiform-local-lean}), which is the key technical tool used in the reductions. 
We used Aristotle \cite{achim2025aristotleimolevelautomatedtheorem} to formalize the proof of \cref{lem:nonuiform-local-lean} in Lean 4,\footnote{\url{https://github.com/xuyinzhan/max-ip-lean}.} which depends on Aleph Prover's formalization \cite{erdos_unit_distance_formalization_2026} of \cite{UnitDistance}. The formalization assumes three standard results from algebraic number theory. 
See \cref{app:lean} for more details. 

The authors have validated and edited the proof and are solely responsible for its correctness.

\section{Preliminaries}

Given a field $K$, let $K^{\times}$ denote its multiplicative group.
Given a real number $x \in \R$, we let $\lfloor x \rfloor$ denote its floor, $\lceil x \rceil$ its ceiling, and $\lfloor x \rceil$ its nearest integer.
For a complex number $z$, let $\Re(z)$ denote its real part and $\Im(z)$ its imaginary part. Throughout the paper, we use $\log$ to denote natural logarithm. 

We define constructible dimension functions. 
Formally, to ensure that our reduction is efficient, we require that the number of dimensions in the constructed $\Z$-OV instance can be computed efficiently.

\begin{definition}[Constructible dimension function]\label{def:constructible}
    A function $D:\Z_{\ge1}\to\Z_{\ge1}$ is constructible if $D(n)$ can be computed exactly in $n^{o(1)}$ time.
\end{definition}

We use the following standard tools from number theory.

\begin{theorem}[Prime Number Theorem]
    \label{thm:prime-number}
    Let $\pi(n)$ denote the number of primes $\leq n$.
    Then,
    \begin{equation*}
        \lim_{n \rightarrow \infty} \frac{\pi(n)}{n/\log n} = 1 \text{.}
    \end{equation*}
\end{theorem}

\begin{theorem}[Chinese Remainder Theorem]
    \label{thm:chinese-remainder}
    Given $d$ pairwise co-prime integers $q_1,q_2,\dots,q_d$, and $d$ integers $r_1,r_2,\dots,r_d$, there is
    exactly one integer $0 \leq t < \prod_{i = 1}^{d} q_i$ such that
    \begin{equation*}
        t \equiv r_i \pmod{q_i} \text{ for all } i \in [d] \text{.}
    \end{equation*}
    We call this the Chinese remainder representation (or the $\CRR$ encoding) of the $r_i$’s (with respect to these
    $q_i$’s). We also denote
    \begin{equation*}
        t = \CRR(\set{r_i}; \set{q_i})
    \end{equation*}
    for convenience. We sometimes omit the sequence $\set{q_i}$ for simplicity, when it is clear from the context.
    Moreover, $t$ can be computed in polynomial time with respect to the total bits of all the given integers.
\end{theorem}

\section{An Improved Local Reduction to \texorpdfstring{$\Z$}{Z}-OV}
\label{sec:gadget}

The goal of this section is to obtain a better reduction to $\Z$-OV.
We formally state the required reduction below.

\begin{lemma}
    \label{lem:nonuiform-local-lean}
    There exist absolute constants $K_0, L_0, w_0$ such that for all integers $b \geq 1$ and $L \geq L_0$, there exist functions $\alpha, \beta: \set{0, 1}^{b} \rightarrow \Z_{\ge 0}^{w_0}$ and $V \subseteq \Z_{\ge 0}$ satisfying the following properties.
    Let $B = \lceil L^{K_0 b \sqrt{L}} \rceil$. Then:
    \begin{enumerate}
        \item (Entry Bound) For all $a \in \set{0, 1}^{b}$ and $k \in [w_0]$, $\alpha(a)_{k}, \beta(a)_{k} < B$ and for all $v \in V$, $v \leq B$. Also, $|V| \leq B$.
        \item (Product Bound) For all $x, y \in \set{0, 1}^{bL}$, let $x_i, y_i$ denote the $i$-th block of size $b$ in $x, y$ respectively.
        Then
        \begin{equation*}
            \sum_{i = 1}^{L} \alpha(x_i) \cdot \beta(y_i) \leq B
        \end{equation*}
        \item (Correctness) For all $x, y \in \set{0, 1}^{bL}$ and $x_i, y_i$ as above,
        \begin{equation*}
            x \cdot y = 0 \quad\Longleftrightarrow\quad
           \sum_{i = 1}^{L} \alpha(x_i) \cdot \beta(y_i) \in V \text{.}
        \end{equation*}
    \end{enumerate}
\end{lemma}

For clarity, we will occasionally denote $\alpha, \beta, V$ as $\alpha_{b, L}, \beta_{b, L}, V_{b, L}$ to make the dependence on $b, L$ explicit.
\Cref{lem:nonuiform-local-lean} has been formally verified in Lean 4 (see \Cref{app:lean} for details).

\subsection{Number Theory Preliminaries}
\label{sec:number-theory-prelims}

We begin with several standard notions from algebraic number theory, with some simple examples to illustrate the definitions.

\begin{definition}[Number field]
\label{def:number-field}
A \emph{number field} is a finite-dimensional field extension $K/\Q$.
Its degree, denoted $[K:\Q]$, is the dimension of $K$ as a vector space over $\Q$.
\end{definition}

As an example $\Q(\sqrt{2}) = \set{a + b \sqrt{2} \mid a, b \in \Q}$ is a degree-2 extension (also called a quadratic extension).
$\Q(2^{1/3})$ is a degree-3 extension (also called a cubic extension).

\begin{definition}[Real embedding]
An embedding of $K/\Q$ into $\C$ is a homomorphism $\sigma: K \hookrightarrow \C$ that fixes $\Q$.
An embedding is \emph{real} if its image is contained in the real numbers: $\sigma(K)\subseteq \mathbb R$ and \emph{complex} otherwise.
\end{definition}

\begin{definition}[Totally real and totally imaginary field]
A number field $K$ is \emph{totally real} if every embedding $\sigma:K \hookrightarrow\mathbb C$ is real. 
A number field $K$ is \emph{totally imaginary} if every embedding $\sigma:K \hookrightarrow\mathbb C$ is complex.
\end{definition}

For example, $\Q(\sqrt{2})$ is totally real (an embedding must send $\sqrt{2}$ to $\sqrt{2}$ or $-\sqrt{2}$), while $\Q(i)$ is totally imaginary ($i$ must be mapped to $i$ or $-i$).

\begin{definition}[Infinite Place]
An infinite place of $K/\Q$ into $\C$ is an equivalence class of embeddings $\sigma:K\hookrightarrow\C$, where two embeddings $\sigma$ and $\tau$ are equivalent if $\tau=\sigma$ or $\tau=\overline{\sigma}$.
An infinite place is \emph{real} if $\sigma(K) \subseteq \R$ for one (equivalently every) representative embedding $\sigma$ and \emph{complex} otherwise.
\end{definition}

We observe that for a real place, the equivalence class consists of one homomorphism, since $\sigma_1 = \overline{\sigma_2}$ implies $\sigma_1 = \sigma_2$.
For a complex place, the equivalence class consists of two homomorphisms that are complex conjugates of each other.
For a field extension $K/L$, we say an infinite place $w$ of $K$ lies above an infinite place $v$ of $L$ if the restriction to $L$ of any representative of $w$ represents the place $v$.

\begin{lemma}[{\cite[Proposition 2.1(b)]{milne2003fields}}]
    \label{lem:num-embeddings-ub}
    Let $K/\Q$ be a finite extension of degree $[K:\Q]$.
    Then, there are exactly $[K:\Q]$ embeddings of $K$ into $\C$.
\end{lemma}

\begin{definition}[CM field ({\cite[Definition A.4]{UnitDistance}})]
    \label{def:cm-field}
    A \emph{CM field} is a totally imaginary quadratic extension $K/L$ of a totally real field. 
    The nontrivial automorphism of $K/L$, also known as the complex conjugation automorphism, is denoted $c$.
    If $K = L(i)$ with $L$ totally real, then for every complex embedding $\sigma: K \hookrightarrow \C$, $\sigma(c(\alpha)) = \overline{\sigma(\alpha)}$. 
\end{definition}

For example $\Q(i)$ is a CM field and the complex conjugation automorphism maps $i$ to $-i$.
When $K = L(i)$ with $L$ totally real, since $\sigma(c(c(\alpha))) = \sigma(\alpha)$ for all $\alpha$ and $\sigma$ is an automorphism, we have $c(c(\alpha)) = \alpha$ is the identity map.

\begin{definition}[Ring of integers]
Let $K$ be a number field. Its \emph{ring of integers} is
\[
\calO_K = \set{ x\in K : x \text{ satisfies a monic polynomial with coefficients in } \Z}.
\]
\end{definition}

For example, $\calO_{\Q} = \Z$.

\begin{definition}[Ideal]
    Given a ring $R$, a subset $I \subseteq R$ is an ideal if: (i) $a + b \in I$ for all $a, b \in I$ and (ii) $r a \in I$ for all $r \in R$ and $a \in I$.
    An ideal is proper if $I \subsetneq R$.
    An ideal is principal if it is generated by a single element.
\end{definition}

Given two ideals $I, J \subseteq R$, their product $IJ$ is the ideal generated by $\set{ab \mid a \in I, b \in J}$.

\begin{definition}[Prime and Maximal ideals]
    A proper ideal $I \subsetneq R$ is prime if for all $a, b \in R$, $ab \in I$ implies that $a \in I$ or $b \in I$.
    A proper ideal $I \subsetneq R$ is maximal if there is no proper ideal $I'$ satisfying $I \subsetneq I' \subsetneq R$.
\end{definition}

For an ideal $I$, we say that a prime ideal $P$ divides $I$  if $I \subseteq P$.
Since rings of integers of number fields are Dedekind Domains (see e.g. Theorem 3.29 of \cite{milne2020algebraic}), every ideal factors uniquely into a product of prime ideals.
We use the following standard fact regarding prime ideals of the ring of integers.

\begin{fact}[Theorem 3.1 of \cite{neukirch2013algebraic}]
    \label{fact:prime-ideal-maximal}
    Every non-zero prime ideal $\mathfrak p$ of $\calO_{K}$ is maximal.
    Since quotients by maximal ideals are fields, $\calO_{K}/\mathfrak{p}$ is a field.
\end{fact}

To distinguish prime ideals and prime numbers, we say a prime $q$ is rational if $q \in \Z$.
In a number field $K$, the ideal generated by $q$ in $\calO_{K}$ is denoted $q \calO_{K}$. 
Note that $q \calO_{K}$ is not necessarily a prime ideal in $\calO_{K}$.
For an extension $K/L$ and prime ideal $\frakp \subseteq \calO_{L}$, observe that if a prime ideal $\frakP \subseteq \calO_{K}$ lies above the ideal generated by $\frakp$, denoted $\frakp \calO_{K}$, then since $\frakp \subseteq \frakP \cap \calO_{L}$, $\frakP \cap \calO_{L}$ is prime and therefore equal to $\frakp$.

\begin{definition}[Local ring]
Let $K$ be a number field and let $\mathfrak p$ be a prime ideal of
$\mathcal O_K$. The \emph{local ring of $K$ at $\mathfrak p$} is
\[
\mathcal O_{K,\mathfrak p}
:=
\left\{
\frac{a}{b}
:
a,b\in\mathcal O_K,\;
b\notin\mathfrak p
\right\} \text{.}
\]
\end{definition}

\begin{definition}[Residue field]
If $\mathfrak p$ is a prime ideal of $\mathcal O_K$, the field $\calO_K/\mathfrak p$ is the \emph{residue field} of $\mathfrak p$.
\end{definition}

From Chapter 1 Corollary 11.2 of \cite{neukirch2013algebraic}, we have for any prime (and therefore maximal) ideal $\frakp$, an isomorphism $\calO_{K}/\frakp \simeq \calO_{K, \frakp}/\frakp \calO_{K, \frakp}$.

\begin{definition}[Residue degree]
Let $K$ be a number field and let $q$ be a rational prime.
For a prime ideal $\mathfrak p$ above $q$, the \emph{residue degree} of $\mathfrak p$ over $q$ is $f(\mathfrak p\mid q) := [\mathcal O_K/\mathfrak p : \mathbf F_q]$.
\end{definition}

\begin{definition}[Splitting completely (as in Definition A.2 of \cite{UnitDistance})]
Let $q$ be a rational prime. 
We say that $q$ is \emph{unramified} in $K$ if
\[
q\mathcal O_K = \prod_{i = 1}^{g} \mathfrak{p}_{i}^{e(\frakp_i|q)} \text{.}
\]
where the prime ideals $\mathfrak p_i$ are distinct and the exponents (called ramification indices) $e(\frakp_{i} | q) = 1$.
We say $q$ \emph{splits completely} if it further satisfies $g = [K:\Q]$ or equivalently $f(\frakp_i | q) = 1$ i.e. each residue field is
\[
\mathcal O_K/\mathfrak p_i \cong \mathbb F_q \text{.}
\]
\end{definition}

As an example, $7$ splits completely in $\Q(\sqrt{2})$, since $(3 + \sqrt{2})(3 - \sqrt{2}) \equiv 0 \pmod{7}$.
Furthermore, for the prime ideal $(7, 3 - \sqrt{2})$, we can see that any element in $\calO_{K}$ is $a + b \sqrt{2} \equiv a + 3 b \pmod{(7, 3 - \sqrt{2})}$ and thus an element of $\F_{7}$.
The following definition will only be used in \Cref{sec:arithmetic} but we state it here as it depends on the above definitions.

\begin{definition}[Definition A.2 of \cite{UnitDistance}]
    \label{def:unramified}
    An \emph{infinite ramification} of an extension $K/F$ is a real place $v: F \rightarrow \C$ such that there is a complex place $w: K \rightarrow \C$ lying above $v$.
    An extension $K/F$ is \emph{everywhere unramified} if there is no finite ramification (i.e. all finite primes of $F$ are unramified) and no infinite ramification.
\end{definition}

We require some additional preliminaries regarding prime ideals.

\begin{definition}[Conjugate prime ideals]
    \label{def:conjugate-prime}
Let $c$ be the complex conjugation automorphism of a CM field $K$.
If $\mathfrak p$ is a prime ideal of $\mathcal O_K$, then $c \mathfrak p = \set{c(x) \mid x\in\mathfrak p}$ is also a prime ideal.
The pair $\set{\mathfrak p,c\mathfrak p}$ is called a \emph{conjugate pair of prime ideals}.
\end{definition}

The fact that prime ideals are preserved under automorphism is standard in commutative algebra (see e.g. Chapter 1 of \cite{neukirch2013algebraic} or \cite{atiyah2018introduction}).

\begin{definition}[Valuation at a prime ideal (see e.g. Chapter 3 of \cite{milne2020algebraic})]
    \label{def:valuation}
    Let $\mathfrak p$ be a prime ideal of $\mathcal O_K$.
    For a nonzero element $x\in K^\times$, the valuation $\operatorname{val}_{\mathfrak p}(x)$ is the exponent of $\mathfrak p$ appearing in the factorization of the principal ideal $(x)$.
\end{definition}

For example, if $q$ splits completely, then $\val_{\frakp}(q) = 1$ for all $\frakp$ dividing $q$ and $0$ otherwise.
Note that valuation satisfies: 
\[\val_{\frakp}(xy) = \val_{\frakp}(x) + \val_{\frakp}(y)\text{, } \val_{\frakp}(x + y) \geq \min(\val_{\frakp}(x), \val_{\frakp}(y))\text{.}\]
We will require the following standard lemma.

\begin{lemma}
    \label{lem:valuation-complex}
    Let $K$ be a CM field with complex conjugation automorphism $c$ and $\frakp$ a prime in $\calO_{K}$.
    Then, for any $z \in K^{\times}$,
    \begin{equation*}
        \val_{\frakp}(c(z)) = \val_{c \frakp}(z) \text{.}
    \end{equation*}
\end{lemma}

\begin{proof}[Proof of \Cref{lem:valuation-complex}]
    Consider the unique factorization of $(z)$ into prime ideals:
    \begin{equation*}
        (z) = \prod_{\frakp_i} \frakp_i^{\val_{\frakp_i}(z)} \text{.}
    \end{equation*}
    Applying the automorphism $c$ we have
    \begin{equation*}
        (c(z)) = c((z)) = \prod_{\frakp_i} (c \frakp_i)^{\val_{\frakp_i}(z)}
    \end{equation*}
    where each $(c \frakp_i)$ is a prime ideal.
    In particular, we have $\val_{c\frakp_i}(c(z)) = \val_{\frakp_i}(z)$. 
    Observe that $c = c^{-1}$ (see discussion below \Cref{def:cm-field}).
    Setting $\frakp := (c\frakp_i)$, we have $\val_{\frakp}(c(z)) = \val_{c\frakp}(z)$, for any $\frakp$ with nonzero valuation in $(c(z))$.
    Otherwise, if $\frakp$ has zero valuation in $(c(z))$, then $c \frakp$ also has zero valuation in $(z)$, preserving the equality.
\end{proof}

\subsection{Proof of \texorpdfstring{\Cref{lem:nonuiform-local-lean}}{Lemma~\ref*{lem:nonuiform-local-lean}}}

We will prove \Cref{lem:nonuiform-local-lean} assuming the following \Cref{prop:codebook}, with all necessary preliminaries introduced above.
This section contains a proof of \Cref{lem:nonuiform-local-lean} assuming \Cref{prop:codebook}.
See \Cref{sec:arithmetic} for a proof of \Cref{prop:codebook}.

\begin{restatable}[Medium-Prime Norm-One Codebook]{proposition}{NumberTheoryProp}
\label{prop:codebook}
    There exist absolute constants $C_0>0$ and $T \geq 16$ such that the following holds for all sufficiently large $L$ and $b$.  Put $Y:=\ceil{4\sqrt L}$.
    There exist:
    \begin{enumerate}[label=(\roman*)]
        \item a CM field $K$ with complex conjugation $c$ and degree $[K:\Q]\le C_0b$;
        \item $T$ distinct rational primes $q_1,\ldots,q_T\in [Y,2Y]$, each splitting completely in $K$, and $Q:=\prod_{t=1}^Tq_t$;
        \item for each block bit $j\in[b]$, two prime ideals
        \[
              \mathfrak p_{j,1},\mathfrak p_{j,2}\subseteq \OO_K
        \]
        lying above two distinct rational primes $q_{j,1},q_{j,2}\in\{q_1,\ldots,q_T\}$, with $q_{j,1}q_{j,2}>L$;
        \item for every $a\in\{0,1\}^b$, an element $u_a\in K^\times$;
    \end{enumerate}
    with the following properties.  For every $a\in\{0,1\}^b$,
    \begin{align}
       u_ac(u_a)&=1, \nonumber \\
       |\sigma(u_a)|&=1 \quad\text{for every complex embedding }\sigma:K\hookrightarrow\C, \nonumber \\
       Q^2u_a&\in\OO_K. \label{eq:denom}
    \end{align}
    Moreover, for every selected prime ideal $\mathfrak p_{j,r}$,
    \begin{align}
       \val_{\mathfrak p_{j,r}}(u_a)&=2a_j \label{eq:val1} \\
       \val_{c\mathfrak p_{j,r}}(u_a)&=-2a_j \text{.} \label{eq:val2}
    \end{align}
    Because each selected rational prime splits completely in $K$, one also has
    \[
            \OO_K/\mathfrak p_{j,r}\cong \F_{q_{j,r}}.
    \]
\end{restatable}

Fix parameters $b, L$ with $L$ sufficiently large and let $K, q_1, \dots, q_T$ be guaranteed by \Cref{prop:codebook}. 
Let
\[
     E_t:=q_t-1\qquad (t=1,\ldots,T).
\]
Since $q_t = \Theta(\sqrt{L})$ and $T$ is constant,
\begin{equation}\label{eq:Qbound}
      Q=\prod_{t=1}^Tq_t=L^{O(1)},\qquad
      Q^{O(E_t)}=\exp(O(\sqrt L\log L)).
\end{equation}

For a block $a\in\{0,1\}^b$, let $u_{a}$ be guaranteed by \Cref{prop:codebook} and define
\begin{equation*}
      F_t(a):=Q^{2E_t}\,c(u_a^{E_t})\in K.
\end{equation*}
By \eqref{eq:denom}, we have $Q^2 u_a \in \calO_{K}$ and therefore $c(Q^2 u_a) = Q^2 c(u_a) \in \calO_{K}$ since $c$ is an embedding fixing $\Q$ and therefore $Q^2$. 
In particular, $F_t(a) = Q^{2 E_t} c(u_a)^{E_t} = (Q^2 c(u_a))^{E_t} \in\OO_K$.
For block vectors
\[
      x=(x_1,\ldots,x_L),\qquad y=(y_1,\ldots,y_L),
      \qquad x_i,y_i\in\{0,1\}^b,
\]
define
\begin{equation*}
      K_t(x,y):=\sum_{i=1}^L F_t(x_i)F_t(y_i)\in\OO_K \text{.}
\end{equation*}

\begin{lemma}[Valuation Table]\label{lem:valuation-table}
Let $\mathfrak p_{j,r}$ lie above the rational prime $q_t$.  Then, for every block $a$,
\[
      \val_{\mathfrak p_{j,r}}(F_t(a))=2E_t(1-a_j).
\]
Consequently, for any blocks $a, b$,
\[
      \val_{\mathfrak p_{j,r}}(F_t(a)F_t(b)) = 2E_t(2-a_j-b_j).
\]
In particular, this product has valuation $0$ exactly when $a_j=b_j=1$, and has valuation at least $2E_t$ otherwise.
\end{lemma}

\begin{proof}
Since the rational prime $q_t$ splits completely in $K$, it is unramified; hence $\val_{\mathfrak p_{j,r}}(q_t)=1$. 
For any other rational prime $q$ dividing $Q$, we claim $\val_{\frakp_{j,r}}(q) = 0$ so that $\val_{\mathfrak p_{j,r}}(Q)=1$. 
This follows as $\gcd(q, q_t) = 1$ so that Bezout's identity guarantees integers $a, b$ such that $aq + bq_t = 1$.
In other words, if $q, q_t \in \frakp_{j, r}$, then $\frakp_{j, r}$ is not a proper ideal, a contradiction.
By \eqref{eq:val2} and \Cref{lem:valuation-complex},
\[
     \val_{\mathfrak p_{j,r}}(c(u_a))=-2a_j.
\]
Therefore,
\[
     \val_{\mathfrak p_{j,r}}(F_t(a)) = E_t(2 \val_{\frakp_{j,r}}(Q) + \val_{\frakp_{j,r}}(c(u_a))) = E_t(2- 2 a_j) = 2E_t(1-a_j) \text{.}
\]
The product formula follows by additivity of valuations.
\end{proof}

\begin{lemma}[Multi-Characteristic Anti-Cancellation]\label{lem:anticancel}
For $x,y\in\{0,1\}^{bL}$, written in $L$ blocks of length $b$, the following are equivalent:
\[
     x\cdot y=0
     \quad\Longleftrightarrow\quad
     \text{for all }t\text{ and all selected }\mathfrak p_{j,r}\mid q_t,
     \;\val_{\mathfrak p_{j,r}}(K_t(x,y))\ge 2E_t .
\]
In particular, if $x\cdot y=0$ and $x'\cdot y'>0$, then
\[
      (K_1(x,y),\ldots,K_T(x,y))
      \ne
      (K_1(x',y'),\ldots,K_T(x',y')).
\]
\end{lemma}

\begin{proof}
If $x\cdot y=0$, then for every block $i \in [L]$ and bit $j \in [b]$ we never have $x_{i,j}=y_{i,j}=1$.
By \Cref{lem:valuation-table}, every summand $F_t(x_i)F_t(y_i)$ has valuation at least $2E_t$ at every selected prime over $q_t$. 
Hence the sum $K_t(x,y)$ has valuation at least $2E_t$ at every such prime, since the valuation of the sum is at least the minimum of the valuations (see discussion below \Cref{def:valuation}).

Conversely suppose $x\cdot y>0$.  Then there exists a bit $j$ for which
\[
      k:=\#\{i\in[L]:x_{i,j}=y_{i,j}=1\}
\]
satisfies $1\le k\le L$.  
The bit $j$ has two selected prime ideals $\mathfrak p_{j,1},\mathfrak p_{j,2}$ lying over distinct rational primes $q_{j,1}$ and $ q_{j,2}$ with $q_{j,1}q_{j,2}>L$.  
Therefore $k$ is not divisible by at least one of $q_{j,1}, q_{j,2}$ (otherwise $k \geq q_{j, 1} q_{j, 2}$).
Choose such a prime and write it as $q_t$, with corresponding selected prime ideal $\mathfrak p=\mathfrak p_{j,r}$.

For a bad block, $a_{j}=1$, \Cref{lem:valuation-table} implies that the element $F_t(a)$ has valuation $0$ at $\mathfrak p$.
Since $q_t$ splits completely in $K$ and $\frakp$ lies above $q_t$, we also have $\val_{\frakp}(q_t) = 1$.
Furthermore, since $a_j=1$, \Cref{prop:codebook} gives 
\begin{equation*}
    \val_{\mathfrak p}(q_t^2c(u_a)) = 2 \val_{\frakp}(q_t) + \val_{\frakp}(c(u_a)) = 2 - 2 a_j = 0 \text{.}
\end{equation*}
Next, we claim $q_t^2 c(u_a)$ is in the local ring $\calO_{K, \frakp}$.
Write $Q = q_t Q'$ so that $Q' \not\in \frakp$ (as argued in \Cref{lem:valuation-table}, $\frakp$ does not divide any other rational prime in $Q$) and $(Q')^2 \not\in \frakp$.
Then, $q_t^2 c(u_a) = (Q^2 c(u_a))/(Q')^2 \in \calO_{K, \frakp}$.
Since the valuation of $q_t^2 c(u_a)$ at $\frakp$ is $0$, we also have $q_t^2c(u_a) \in \calO_{K,\mathfrak p}^{\times}$
and has a well-defined nonzero residue modulo $\frakp$.
Reducing modulo $\mathfrak p$ and observing that $\calO_{K,\frakp}/\frakp \calO_{K,\frakp} \simeq \calO_{K}/\frakp \simeq \F_{q_t}$ since $q_t$ splits completely, we have that $q_t^2 c(u_a)$ is a non-zero element in $\F_{q_t}$.
Fermat's Little Theorem then implies that 
\begin{equation*}
    (q_t^2 c(u_a))^{E_t} \equiv (q_t^2 c(u_a))^{q_t - 1} \equiv 1 \pmod{\frakp} \text{.}
\end{equation*}
Also $(Q')^{2E_t} \equiv ((Q')^{E_t})^{2} \equiv 1^2 \equiv 1 \pmod{\mathfrak p}$.  Hence $F_t(a)\equiv1\pmod{\mathfrak p}$, and similarly $F_t(b)\equiv1\pmod{\mathfrak p}$ for the right block $b$ when $b_j=1$.  Thus every bad product satisfies
\[
      F_t(a)F_t(b)\equiv1\pmod{\mathfrak p}.
\]
The other summands are divisible by $\mathfrak p$, since the valuation is at least $2 E_t$.
Therefore
\[
      K_t(x,y)\equiv k\pmod{\mathfrak p}.
\]
Since the residue characteristic is $q_t$ and $q_t\nmid k$, this residue is nonzero.  Hence
\[
      \val_{\mathfrak p}(K_t(x,y))=0<2E_t.
\]
This proves the equivalence and the separation of orthogonal from nonorthogonal algebraic value vectors.
\end{proof}

\subsubsection{Constructing the Reduction}

We now proceed to describe the reduction.
First, we embed binary vectors into the complex vectors with entries in $\C$.
Fix one complex embedding $\sigma:K\hookrightarrow\C$.  For $a\in\{0,1\}^b$, set
\[
      z_t(a):=\sigma(F_t(a)).
\]
Then
\[
      \sigma(K_t(x,y))=\sum_{i=1}^L z_t(x_i)z_t(y_i).
\]
Because $|\tau(u_a)|=1$ for every embedding $\tau$, \eqref{eq:Qbound} gives an absolute constant $C_2$ such that
\begin{equation}\label{eq:blockarch}
      |\tau(F_t(a))|\le \exp(C_2\sqrt L\log L)
\end{equation}
for every $t,a,\tau$.

\begin{lemma}[Product-Formula Separation]\label{lem:product-separation}
There is an absolute constant $C_3$ such that if
\[
      (K_1(x,y),\ldots,K_T(x,y))\ne (K_1(x',y'),\ldots,K_T(x',y')),
\]
then, for some $t$,
\[
      |\sigma(K_t(x,y)-K_t(x',y'))|
      \ge \exp(-C_3b\sqrt L\log L).
\]
\end{lemma}

\begin{proof}
Choose $t$ with $\Delta:=K_t(x,y)-K_t(x',y')\ne0$.  By construction $\Delta\in\OO_K$.  For every embedding $\tau$,
\[
      |\tau(\Delta)|\le 2L\exp(2C_2\sqrt L\log L)
      \le \exp(C_4\sqrt L\log L)
\]
for large $L$ after enlarging $C_4$.  

We now define the field norm, $\Norm_{K/\Q}(x) := \prod_{\tau: K \hookrightarrow \C} \tau(x)$, where the product is taken over all embeddings $\tau$. 
From \Cref{lem:num-embeddings-ub}, there are $[K:\Q]$ such embeddings.
We claim that $\Norm_{K/\Q}(x) \in \Z$ for any $x \in \calO_{K}$.
Since $x \in \calO_{K}$, $x$ satisfies a monic polynomial with coefficients in $\Z$, e.g. $x^n + a_{n - 1} x^{n - 1} + \dots + a_0 = 0$.
Let $m_{x}$ denote the minimal polynomial with degree $n$.
We claim that each root of $m_{x}$ appears equally often in the set $\set{\tau(x)}$.
In particular, the product of $\tau(x)$ over all $\tau$ is $(-1)^{w} a_0^{v}$ for some integers $w, v$, and therefore an integer.
To see this, let \(\alpha\) be a root of \(m_x\).
By Proposition 2.1(b) of \cite{milne2003fields}, \(\alpha\) determines an embedding
\[
\rho_\alpha:\mathbb Q(x)\hookrightarrow \mathbb C,\qquad x\mapsto \alpha \text{.}
\]
Since \(K=\mathbb Q(x)(\gamma)\) (see e.g. Theorem 5.1 of \cite{milne2003fields} and \Cref{def:separable} for the required definition of separable) for some \(\gamma\in K\), applying Proposition 2.12 of \cite{milne2003fields} to the extension \(K/\mathbb Q(x)\) shows that \(\rho_\alpha\) extends to exactly $[K:\mathbb Q(x)]$ embeddings \(K\hookrightarrow \mathbb C\). 
Hence each root \(\alpha\) of \(m_x\) occurs exactly \([K:\mathbb Q(x)]\) times among the values \(\tau(x)\).

Since $\Norm_{K/\Q}(\Delta)$ is a nonzero integer,
\[
      1 \le |\Norm_{K/\Q}(\Delta)|
      =\prod_{\tau:K\hookrightarrow\C}|\tau(\Delta)|.
\]
Here the product is over all $[K:\Q]$ complex embeddings, counted individually rather than by conjugate pairs. 
Isolating the fixed embedding $\sigma$ and bounding every other factor by the displayed Archimedean upper bound gives
\[
      |\sigma(\Delta)|
      \ge \exp(-C_4([K:\Q]-1)\sqrt L\log L).
\]
The degree of $K$ is $O(b)$ by \Cref{prop:codebook}, so
\[
      |\sigma(\Delta)|\ge \exp(-C_3b\sqrt L\log L)
\]
for a suitable absolute constant $C_3$.
\end{proof}

We now map the binary vectors into $\Z$-vectors.

\begin{proposition}[Nonuniform Finite Integer Block Gadget]\label{prop:int-gadget}
There are absolute constants $C_5$ and $L_*\ge2$ such that, for every $b\ge1$ and every $L\ge L_*$, there exist two blockwise maps
\[
      \Phi_A,\Phi_B:\{0,1\}^{bL}\to\Z^r,
      \qquad r=O(L),
\]
and a finite set $V_{b,L}\subseteq\Z$ such that
\begin{equation}
    \label{eq:gadgetpred}
      x\cdot y=0
      \quad\Longleftrightarrow\quad
      \Phi_A(x)\cdot\Phi_B(y)\in V_{b,L}.
\end{equation}
Here blockwise means that there are local maps $\alpha_{b,L},\beta_{b,L}:\{0,1\}^b\to\Z^{w}$, with $w=O(1)$, such that
\[
      \Phi_A(x_1,\ldots,x_L)=(\alpha_{b,L}(x_1),\ldots,\alpha_{b,L}(x_L))
\]
and similarly for $\Phi_B$.  Moreover,
\begin{align*}
      \|\Phi_A(x)\|_\infty,\;\|\Phi_B(y)\|_\infty,
      \;\max_{v\in V_{b,L}}|v| &\le 2^{C_5 b\sqrt L\log L},  \\
      |V_{b,L}|&\le 2^{C_5 b\sqrt L\log L} \text{.}
\end{align*}
\end{proposition}

\begin{proof}
Let $b_*$ be the finite threshold in \Cref{prop:codebook}.  Choose once and for all an absolute $L_*$ dominating the $L$-threshold in \Cref{prop:codebook}, the prime-number-theorem threshold needed in the finite fallback below, and the finitely many lower bounds on $L$ used in the rounding inequalities.  We prove the claim for $L\ge L_*$. 
We split into two cases.

\paragraph{Case 1: $b < b_*$.}
If $b<b_*$, we use the Chinese Remainder Theorem (CRT) to directly construct the required maps.
Since $b<b_*$ is constant, the Prime Number Theorem implies that the interval $(L,2L]$ contains at least $b_*$ primes.  Choose $b$ distinct primes
\[
      p_1,\ldots,p_b\in(L,2L].
\]
Then
\[
      P:=\prod_{j=1}^b p_j\le (2L)^b\le L^{Cb}
\]
for an absolute constant $C$ depending only on $b_*$.  For a block $a\in\{0,1\}^b$, let $t(a)\in[0,P)$ be the CRR encoding with $t(a) \equiv a_j \pmod{p_j}$ for all $j \in [b]$.
Put one local coordinate on each side, $\alpha(a)=\beta(a)=t(a)$. 
Then for full block vectors,
\[
       D(x,y):=\sum_{i=1}^L t(x_i)t(y_i)
\]
satisfies $D(x,y)\equiv\sum_i x_{i,j}y_{i,j}\pmod {p_j}$.  For each fixed $j$, the integer
\[
      s_j:=\sum_{i=1}^L x_{i,j}y_{i,j}
\]
lies in $[0,L]$.  Because $p_j>L$, the congruence $s_j\equiv0\pmod {p_j}$ is therefore equivalent to $s_j=0$.  Hence
\[
      D(x,y)\equiv0\pmod P
      \quad\Longleftrightarrow\quad
      s_j=0\text{ for every }j\in[b]
      \quad\Longleftrightarrow\quad
      x\cdot y=0.
\]
Taking
\[
       V=\{z\in[0,LP^2]\cap\Z:z\equiv0\pmod P\}
\]
gives coordinate, value, and cardinality bounds $L^{O(b)}$, which are stronger than the claimed $2^{O(b\sqrt L\log L)}$ for sufficiently large $L$.  Hence we may assume $b\ge b_*$ and use the arithmetic construction.

\paragraph{Case 2: $b \geq b_*$.}

Let
\begin{equation}
    \label{eq:N-choice}
      N:=\left\lceil\exp(C_6b\sqrt L\log L)\right\rceil
\end{equation}
with $C_6$ sufficiently large.  For every $t$ and block $a$, define integers
\[
      X_t(a):=\lfloor N\Re(z_t(a))\rceil,
      \qquad
      Y_t(a):=\lfloor N\Im(z_t(a))\rceil .
\]
For two blocks $a,b'$ (we use $b'$ to avoid confusion with the parameter $b$), define
\begin{align*}
      R_t(a,b')&:=X_t(a)X_t(b')-Y_t(a)Y_t(b'),\\
      I_t(a,b')&:=X_t(a)Y_t(b')+Y_t(a)X_t(b').
\end{align*}
These are rounded versions of $N^2\operatorname{Re}(z_t(a)z_t(b'))$ and $N^2\operatorname{Im}(z_t(a)z_t(b'))$.  Indeed, write
\[
      X_t(a)=N\operatorname{Re}(z_t(a))+\xi_a,
      \qquad
      Y_t(a)=N\operatorname{Im}(z_t(a))+\eta_a,
\]
with $|\xi_a|,|\eta_a|\le1/2$, and similarly for $b'$.  Using \eqref{eq:blockarch}, for each block pair

\begin{align}
 |R_t(a,b')-N^2\operatorname{Re}(z_t(a)z_t(b'))|&\le C_7N\exp(C_2\sqrt L\log L)+C_7,\label{eq:block-round-error}\\
 |I_t(a,b')-N^2\operatorname{Im}(z_t(a)z_t(b'))|&\le C_7N\exp(C_2\sqrt L\log L)+C_7.\nonumber
\end{align}

Let $M_0$ be an integer larger than twice the absolute value of every possible component sum
\[
      \sum_{i=1}^L R_t(x_i,y_i),
      \qquad
      \sum_{i=1}^L I_t(x_i,y_i).
\]
By \eqref{eq:blockarch}, \eqref{eq:block-round-error}, and the choice of $N$ \eqref{eq:N-choice}, we may take
\[
      M_0\le 2^{C_8b\sqrt L\log L}
\]
for a sufficiently large $C_8$.

For each block $a$, include, for every $t=1,\ldots,T$, the four left local coordinates
\[
      M_0^{2t}X_t(a),\quad
      M_0^{2t}Y_t(a),\quad
      M_0^{2t+1}X_t(a),\quad
      M_0^{2t+1}Y_t(a),
\]
and the four right local coordinates
\[
      X_t(b'),\quad -Y_t(b'),\quad Y_t(b'),\quad X_t(b').
\]
Concatenating these local maps over the $L$ blocks gives $r=4TL=O(L)$ coordinates.  For block vectors $x,y$, their integer dot product is
\begin{equation*}
   \Phi_A(x)\cdot\Phi_B(y)
   =\sum_{t=1}^T M_0^{2t}\sum_{i=1}^L R_t(x_i,y_i)
    +\sum_{t=1}^T M_0^{2t+1}\sum_{i=1}^L I_t(x_i,y_i).
\end{equation*}

The packing is injective on component sums. If two packed sums are equal, subtract them and divide by the smallest used power of $M_0$, namely $M_0^2$, and then reindex the remaining exponents from $0$.  We get an identity $\sum_s D_sM_0^s=0$.  Each $D_s$ is the difference of two possible component sums, so the choice of $M_0$ gives $|D_s|<M_0$.  Reducing modulo $M_0$ gives $D_0=0$; dividing by $M_0$ and repeating gives every $D_s=0$.  Hence equality of packed dot products is equivalent to equality of all real and imaginary component sums.

By \Cref{lem:product-separation}, an orthogonal algebraic value vector and a nonorthogonal algebraic value vector differ in some complex component by at least
\[
      \delta_0:=\exp(-C_3b\sqrt L\log L).
\]
Therefore either the corresponding real parts or imaginary parts differ by at least $\delta_0/\sqrt2$.  Summing \eqref{eq:block-round-error} over the $L$ blocks, the total rounding error in any component sum is at most
\[
      L\bigl(C_7N\exp(C_2\sqrt L\log L)+C_7\bigr).
\]
Choosing $C_6$ sufficiently large (recall this ensures $N$ is sufficiently large) ensures
\[
      N^2\delta_0/\sqrt2
      >2L\bigl(C_7N\exp(C_2\sqrt L\log L)+C_7\bigr).
\]
Indeed, after dividing the dominant right-hand term by $N$, the comparison is between exponents
\[
      (C_6-C_3)b\sqrt L\log L
      \quad\text{and}\quad
      C_2\sqrt L\log L+O(\log L).
\]
Since $b\ge b_*\ge1$ in the arithmetic case, one fixed choice of $C_6>C_3+C_2+O(1)$ works for all large $L$; the remaining additive term $2LC_7$ is then smaller as well.
Thus after rounding, no orthogonal pair and nonorthogonal pair have the same packed dot product.  Define
\[
      V_{b,L}:=\{\Phi_A(x)\cdot\Phi_B(y):x\cdot y=0\}.
\]
Then \eqref{eq:gadgetpred} follows.  The coordinate and value bounds follow from \eqref{eq:blockarch}, the choices of $N$ and $M_0$, and the fact that $T$ is constant.  Since all possible dot products lie in an interval of length $2^{O(b\sqrt L\log L)}$, the same bound holds for $|V_{b,L}|$.
\end{proof}

We are now ready to prove \Cref{lem:nonuiform-local-lean}.
From \Cref{prop:int-gadget}, we simply need to shift the coordinates so that all values are nonnegative.

\begin{proof}[Proof of \Cref{lem:nonuiform-local-lean}]
Let $L_0:=L_*$, where $L_*$ is the absolute threshold from \Cref{prop:int-gadget}.  
Assume $L\ge L_0$.
The (possibly negative) coordinates in the vectors produced by \Cref{prop:int-gadget} have magnitude at most $2^{O(b\sqrt L\log L)}=L^{O(b\sqrt L)}$.  
Let $V^{\rm old}_{b,L}$ be the old value set, let $r_s$ be the full dimension of the vectors (i.e. $r_s = 4 T L$), and choose an integer $B\ge1$ that bounds both the magnitude of every signed coordinate and $|v|$ for every $v\in V^{\rm old}_{b,L}$.  This still has $B=L^{O(b\sqrt L)}$ by \Cref{prop:int-gadget}.  Replace each signed left coordinate $a_i$ and right coordinate $b_i$ by the nonnegative pair
\[
        (B+a_i,B-a_i),\qquad (B+b_i,B-b_i).
\]
The contribution of this pair is
\[
        (B+a_i)(B+b_i)+(B-a_i)(B-b_i)=2B^2+2a_ib_i.
\]
Therefore the full transformed dot product equals
\[
        2r_sB^2+2D_{\rm old}(x,y),
\]
where $D_{\rm old}$ is the signed dot product from \Cref{prop:int-gadget}.  Since each signed product has magnitude at most $B^2$ and there are $r_s$ signed coordinates,
\[
        |D_{\rm old}(x,y)|\le r_sB^2
\]
for every pair $x,y$.  Define the new value set explicitly by
\[
        V^{+}_{b,L}:=\{2r_sB^2+2v:v\in V^{\rm old}_{b,L}\}.
\]
The choice of $B$ makes every element of $V^{+}_{b,L}$ nonnegative, and the displayed bound also makes every transformed dot product nonnegative.  Moreover every transformed dot product is at most $4r_sB^2$.  Since $r_s=O(L)$, $b\ge1$, and $L\ge2$, this is still $L^{O(b\sqrt L)}$; the same bound holds for the shifted values and the transformed coordinates.  Membership in the value set is preserved exactly.  This doubles the local width and changes the exponent only by an absolute constant.  Enlarging $K_0$ bounds the coordinates, the values in $V^{+}_{b,L}$, its cardinality, and all possible dot products by $\ceil{L^{K_0b\sqrt L}}$.
\end{proof}

\section{\texorpdfstring{$\Z$-OV}{Z-OV} Lower Bound via Uniform Local Reductions}
\label{sec:uniform}

At this point, \Cref{lem:nonuiform-local-lean} already gives a nonuniform reduction from (Boolean) OV to $\Z$-OV.
We can, for example, modify Theorem B.3 of \cite{DBLP:journals/toc/Chen20} to obtain a uniform reduction, and therefore the desired lower bound.
 The setting here is slightly different from \cite{DBLP:journals/toc/Chen20}'s setting. 
One key difference is that the reduction here is for the blockwise bichromatic setting, where \cite{DBLP:journals/toc/Chen20} assumes a single map for both sets of vectors. For completeness, we include a proof in this section, which closely follows \cite{DBLP:journals/toc/Chen20}'s approach.

First, we restate \Cref{lem:nonuiform-local-lean} in the language of a nonuniform local reduction (similar to \cite{DBLP:journals/toc/Chen20}).

\begin{restatable}[Local Bichromatic Nonnegative Reduction]{definition}{LocalReduction}
    \label{def:local-red}
    A local bichromatic nonnegative $(b,L,\kappa)$-reduction of width $w=O(1)$ consists of maps
    \[
           \alpha,\beta:\{0,1\}^b\to\{0,1,\ldots,\ceil{L^{\kappa b}}-1\}^{w}
    \]
    and a set $V\subseteq\Z_{\ge0}$ such that
    \[
           |V|\le \ceil{L^{\kappa b}},\qquad
           \max_{v\in V}v\le \ceil{L^{\kappa b}},
    \]
    every possible full dot product is at most $\ceil{L^{\kappa b}}$, and for all block vectors $x=(x_1,\ldots,x_L)$ and $y=(y_1,\ldots,y_L)$,
    \[
           x\cdot y=0
           \quad\Longleftrightarrow\quad
           \sum_{i=1}^L\alpha(x_i)\cdot\beta(y_i)\in V.
    \]
    It is uniform if $\alpha,\beta$ can be evaluated and $V$ can be generated in time $L^{O(\kappa b)}\poly(b)$.
\end{restatable}

It is easy to see that \Cref{lem:nonuiform-local-lean} implies the following.

\begin{restatable}{corollary}{LocalReductionLemma}
    \label{cor:nonuniform-local}
    There are absolute constants $K_0$ and $L_0\ge2$, and a fixed width $w_0=O(1)$, such that, for every $b\ge1$ and every $L\ge L_0$, there is a (possibly nonuniform) local bichromatic nonnegative $(b, L, K_0 \sqrt{L})$-reduction of width $w_0$, consisting of maps $\alpha_{b, L}, \beta_{b, L}$ and a set $V_{b, L}$.
\end{restatable}

\begin{lemma}[CRT Composition for Local Two-Map Reductions]\label{lem:crt-composition}
There is an absolute constant $C_{\rm crt}$ with the following property.  
Suppose $\kappa\ge1$ and there is a local bichromatic nonnegative $(m,L,\kappa)$-reduction (called the micro-reduction) of constant width $w$. 
Let $B\ge m$ and put $q:=\ceil{B/m}$.  If $q\le L^{\kappa m}$, then there is a local bichromatic nonnegative $(B,L,C_{\rm crt}\kappa)$-reduction of the same constant width.  
If the micro-reduction is given uniformly, the composed reduction can be evaluated and its value set generated in time
\[
       L^{O(\kappa B)}\poly(B)
\]
plus the corresponding micro-reduction construction time.
\end{lemma}

\begin{proof}
Pad a $B$-bit block with zeros and view it as $q$ chunks of length-$m$ blocks.  Let $(\alpha,\beta,V)$ be the micro-reduction.  By \Cref{def:local-red}, all values of $V$ and all possible micro full dot products are nonnegative and at most
\[
       D_{\rm micro}:=\ceil{L^{\kappa m}}.
\]
If $D_{\rm micro}$ is below a fixed absolute threshold $D_0$, then $q\le D_{\rm micro}\le D_0$.  For each of the finitely many pairs $(D_{\rm micro},q)$ with $1\le q\le D_{\rm micro}\le D_0$, hardwire the first $q$ primes larger than $D_{\rm micro}$ into the uniform construction.  Their product is bounded by an absolute constant, which is absorbed in $L^{O(\kappa B)}$ after enlarging $C_{\rm crt}$.  Thus we may assume $D_{\rm micro}$ is large.  Choose $q$ distinct primes $p_1,\ldots,p_q$ in the interval
\[
      [D_{\rm micro}^2,2D_{\rm micro}^2].
\]
For large $D_{\rm micro}$, the prime number theorem gives
\[
      \#\{p\in[D_{\rm micro}^2, 2D_{\rm micro}^2]:p\text{ prime}\}
      =\Theta\left(\frac{D_{\rm micro}^2}{\log D_{\rm micro}}\right),
\]
which is larger than $q$ because $q \le D_{\rm micro}$.  
Thus such a choice is possible, and each selected prime satisfies $p_h > D_{\rm micro}$.  
Put $P:=\prod_{h=1}^q p_h$.  
Since every $p_h \le 2D_{\rm micro}^2$,
\[
       P\le (2D_{\rm micro}^2)^q\le L^{O(\kappa mq)}=L^{O(\kappa B)}.
\]
To find these primes deterministically, scan the interval $[D_{\rm micro}^2,2D_{\rm micro}^2]$ and run a primality test on each integer until $q$ primes have been collected.  The interval length is $D_{\rm micro}^2\le L^{O(\kappa m)}\le L^{O(\kappa B)}$, and each primality test costs polynomial time in $\log D_{\rm micro}=O(\kappa m\log L)$, so this search fits within the final stated time bound.

For a large block $a=(a^{(1)},\ldots,a^{(q)})$ and each local coordinate $r\in[w]$, define $\alpha'(a)_r$ as the unique integer in $[0,P)$ satisfying
\[
       \alpha'(a)_r\equiv \alpha(a^{(h)})_r\pmod {p_h}
       \qquad\text{for all }h\in[q].
\]
Define $\beta'$ analogously from $\beta$.  For full block vectors $x,y$, let
\[
       D'(x,y):=\sum_{i=1}^L\alpha'(x_i)\cdot\beta'(y_i).
\]
For every chunk index $h$,
\[
       D'(x,y)\equiv
       \sum_{i=1}^L \alpha(x_i^{(h)})\cdot\beta(y_i^{(h)})
       \pmod {p_h}.
\]
Because $p_h>D_{\rm micro}$, this residue identifies the integer micro dot product exactly.  Moreover,
\[
      x\cdot y=\sum_{h=1}^q \sum_{i=1}^L x_i^{(h)}\cdot y_i^{(h)},
\]
and every summand is a nonnegative integer because the vectors are Boolean.  Therefore $x\cdot y=0$ if and only if each chunked dot product $\sum_{i=1}^L x_i^{(h)}\cdot y_i^{(h)}$ is zero.  By the micro-reduction property, this is equivalent to requiring that, for every $h$, the identified micro dot product lies in $V$.

The possible range of $D'$ is explicit:
\[
       0\le D'(x,y)\le D_{\rm max}':=LwP^2.
\]
Define
\[
       V':=\{z\in[0,D_{\rm max}']\cap\Z:
              z\bmod p_h\in V\text{ for every }h\in[q]\}.
\]
Then $(\alpha',\beta',V')$ is the desired composed reduction.  The coordinate bound follows from $\alpha'(a)_r,\beta'(a)_r<P\le L^{O(\kappa B)}$, and the full dot-product range is at most $LwP^2\le L^{O(\kappa B)}$.  
For the value-set size, note that $V' \subset [L^{O(\kappa B)}]$.
After increasing $C_{\rm crt}$, all bounds fit the definition of a $(B,L,C_{\rm crt}\kappa)$-reduction.

Uniform evaluation uses $q$ micro-evaluations and CRT reconstruction for each coordinate.  To generate $V'$, enumerate the
\[
       |V|^q\le \ceil{L^{\kappa m}}^q\le L^{O(\kappa mq)}\le L^{O(\kappa B)}
\]
residue tuples and reconstruct their residue classes $r\in[0,P)$ modulo $P$.  For each such $r$, list the integers $r+aP$ with $a\ge0$ and $r+aP\le D_{\rm max}'$.  There are at most $1+D_{\rm max}'/P\le L^{O(\kappa B)}$ representatives per tuple.  Output the generated representatives with duplicates removed, for example by sorting; the total generated list length is $L^{O(\kappa B)}$, so this remains within the stated time.  Arithmetic is on $O(\kappa B\log L)$-bit integers and contributes only polynomial factors.
\end{proof}

\begin{lemma}[Base Search]\label{lem:base-search}
For every fixed $C>0$ and absolute integer width bound $W\ge1$, there are thresholds $L_{\rm bs}=L_{\rm bs}(C,W)$ and $b_{\rm bs}=b_{\rm bs}(C,W)$ with the following property.  Let $L\ge L_{\rm bs}$ and $b\ge b_{\rm bs}$, and assume
\begin{equation}\label{eq:base-range}
       m,L\le C\log\log\log b,\qquad
       1\le\kappa\le (\log\log\log b)^{C}.
\end{equation}
Suppose a local bichromatic nonnegative $(m,L,\kappa)$-reduction of width at most $W$ exists.  Then exhaustive search over widths at most $W$ finds a local bichromatic nonnegative $(m,L,C_{\rm bs}\kappa)$-reduction, for a constant $C_{\rm bs}=C_{\rm bs}(W)$, in time
\[
       b^{o(1)}\le L^{O(\kappa b)}
\]
for the ranges used in \Cref{lem:two-param-uniformization}.
\end{lemma}

\begin{proof}
Choose a sufficiently large absolute constant $C_0$ and put
\[
       R:=\ceil{L^{C_0\kappa m}}.
\]
In the application $W=w_0$ from \Cref{cor:nonuniform-local}.  Enumerate all widths $1\le w\le W$ and all pairs of local maps
\[
       \alpha,\beta:\{0,1\}^m\to\{0,1,\ldots,R-1\}^{w}.
\]
Since a $(m,L,\kappa)$-reduction of some width at most $W$ exists, one of the enumerated candidates contains the witnessing maps after enlarging $C_0$ if necessary.  The number of candidates, summing over $w\le W$, is
\[
       \sum_{w\le W}R^{2w2^m}=L^{O(\kappa m2^m)}=b^{o(1)}
\]
under \eqref{eq:base-range}.  Indeed,
\[
      \log\!\left(L^{O(\kappa m2^m)}\right)=O(\kappa m2^m\log L).
\]
Here $\kappa,m,\log L\le(\log\log\log b)^{O(1)}$ by \eqref{eq:base-range}, while
\[
      2^m\le 2^{O(\log\log\log b)}=(\log\log b)^{O(1)}.
\]
Therefore the logarithm of the search space is bounded by a fixed power of $\log\log b$, which is $o(\log b)$.  Exponentiating gives $L^{O(\kappa m2^m)}=b^{o(1)}$.

For each candidate, compute all one-block contributions
\[
       g(a,b'):=\alpha(a)\cdot\beta(b'),\qquad a,b'\in\{0,1\}^m,
\]
and mark whether $a\cdot b'=0$.  Use dynamic programming over the $L$ block positions to compute two sets of sums: $S_{\rm good}$, obtainable with all block pairs orthogonal, and $S_{\rm bad}$, obtainable with at least one nonorthogonal block pair.  Initially $S_{\rm good}=\{0\}$ and $S_{\rm bad}=\varnothing$.  If $G_{\rm good}$ and $G_{\rm bad}$ are the one-block contribution sets from orthogonal and nonorthogonal block pairs, and $G_{\rm all}:=G_{\rm good}\cup G_{\rm bad}$, update by
\begin{align*}
      S_{\rm good}'&=S_{\rm good}+G_{\rm good},\\
      S_{\rm bad}'&=(S_{\rm bad}+G_{\rm all})\cup(S_{\rm good}+G_{\rm bad}),
\end{align*}
and deduplicate after each step.

Let
\[
       R_{\rm out}:=\ceil{L^{C_1\kappa m}}
\]
with $C_1$ much larger than $C_0$ and $W$.  In the DP, keep only sums in $[0,R_{\rm out}]$ and record an overflow flag; during either update, if a sum in $S_{\rm good}+G_{\rm good}$, $S_{\rm bad}+G_{\rm all}$, or $S_{\rm good}+G_{\rm bad}$ falls outside this interval, set the flag.  Any candidate with the overflow flag set is rejected, so for an accepted candidate no reachable sum has been discarded.  A candidate is accepted only if
\[
       S_{\rm good}\cap S_{\rm bad}=\varnothing,\qquad
       |S_{\rm good}|\le R_{\rm out},\qquad
       \max(S_{\rm good}\cup S_{\rm bad})\le R_{\rm out}.
\]
For an accepted candidate take $V=S_{\rm good}$.  Then all coordinates are at most $R\le R_{\rm out}$ after increasing $C_1$ if necessary, all actual full dot products and all values of $V$ are at most $R_{\rm out}$, and $|V|\le R_{\rm out}$.  Hence the output is a local $(m,L,C_{\rm bs}\kappa)$-reduction for a constant $C_{\rm bs}=C_{\rm bs}(W)$.

Existence of a local $(m,L,\kappa)$-reduction guarantees that at least one candidate passes the test, because its actual orthogonal value set and full dot-product range are bounded by $\ceil{L^{\kappa m}}\le R_{\rm out}$.  The DP states are subsets of $[0,R_{\rm out}]$, so the verification time per candidate is $L^{O(\kappa m)}2^{O(m)}\poly(m)=b^{o(1)}$ in the stated range.  Multiplying by the $b^{o(1)}$ candidate count still gives total search time $b^{o(1)}$.
\end{proof}

\begin{lemma}[Two-Parameter Uniformization]\label{lem:two-param-uniformization}
There are absolute constants $K_1,L_1,B_1$ such that, for every $L\ge L_1$ and $b\ge B_1$ with
\[
       2\le L\le \log\log\log b,
\]
there is a uniform local bichromatic nonnegative $(b,L,K_1\sqrt L)$-reduction.  The maps can be evaluated and the value set can be generated in time
\[
       L^{O(b\sqrt L)}\poly(b)=2^{O(b\sqrt L\log L)}\poly(b).
\]
\end{lemma}

\begin{proof}
By \Cref{cor:nonuniform-local}, for every block length and every $L\ge L_0$ there exists a nonuniform local reduction with parameter $\Lambda:=K_0\sqrt L$, where $L_0,K_0$, and the width bound $w_0$ are absolute.  Fix an absolute constant $C_*$ large enough that the bounds verified below imply \eqref{eq:base-range} with $C=C_*$.  Let $L_1$ dominate $L_0$ and the threshold $L_{\rm bs}(C_*,w_0)$ from \Cref{lem:base-search}.  The constants $K_0,w_0,L_1$ are fixed once and for all and hardwired into the uniform search.  After enlarging $K_0$ if necessary, $\Lambda\ge1$ for $L\ge L_1$.

Set $b_0:=b$.  For $s=1,2,3$, define
\[
       b_s:=\max\left\{1,\left\lceil \frac{\log b_{s-1}}{\Lambda\log L}\right\rceil\right\}.
\]
Then $b_3=O(\log\log\log b)$.  Also $\Lambda=K_0\sqrt L\le(\log\log\log b)^{C_*}$ after increasing $C_*$, so \eqref{eq:base-range} holds for $m=b_3$, $\kappa=\Lambda$, and $C=C_*$.  Let $A:=\log\log\log b$ and $D:=\Lambda\log L$.  Since $2\le L\le A$, we have $D\le C\sqrt A\log A$ for an absolute $C$.  For large $b$, the quantities $\log b_{s-1}/D$ are larger than $1$ for $s=1,2,3$, so the maxima in the definition of $b_s$ are inactive.  Also $\ceil{x}\ge x$, and
\[
       \log (3 D)\le \tfrac12 A
\]
for all large $A$.  Therefore
\[
       b_1\ge \frac{\log b}{2D},\qquad
       b_2\ge \frac{\log b_1}{2D}\ge \frac{\log\log b}{3D},\qquad
       b_3\ge \frac{\log b_2}{2D}\ge \frac{A}{4D}
\]
hold after enlarging a single absolute lower bound on $b$ if necessary.  Indeed, the first bound is weaker than $b_1\ge \log b/D$; the second follows from $\log b_1\ge \log\log b-\log(2D)\ge (2/3)\log\log b$; and the third follows from $\log b_2\ge \log\log\log b-\log(3D)\ge A/2$.  Hence $b_3\ge c\sqrt A/\log A\to\infty$ uniformly for $2\le L\le A$.  Choose the threshold $B_1$ large enough that, for every $b\ge B_1$, all preceding displayed inequalities hold and $b\ge b_{\rm bs}(C_*,w_0)$.  Thus the nonuniform reduction with block length $b_3$, width bound $w_0$, and parameter $\Lambda$ exists, and \Cref{lem:base-search} with $W=w_0$ finds a base reduction with parameter $C_{\rm bs}\Lambda$ in time $b^{o(1)}$.

Compose upward.  For each $s=3,2,1$, use \Cref{lem:crt-composition} to pass from block length $b_s$ to block length $b_{s-1}$.  The number of chunks $q_s:=\ceil{b_{s-1}/b_s}$ satisfies
\[
       q_s\le L^{C\Lambda b_s}
\]
for an absolute constant $C$, by the definition of $b_s$: if $b_s>1$, then
\[
      b_s\ge \frac{\log b_{s-1}}{\Lambda\log L},
\]
so
\[
      b_{s-1}\le e^{\Lambda b_s\log L}=L^{\Lambda b_s}
\]
up to an absolute constant factor coming from the ceiling, and hence
\[
      q_s=\ceil{b_{s-1}/b_s}\le b_{s-1}+1\le L^{C\Lambda b_s}.
\]
If $b_s=1$, then also $b_{s-1}\le e^{\Lambda\log L}=L^{\Lambda}$ by the maximality in the definition of $b_s$, so the same bound holds after enlarging $C$.  Absorbing $C$ into the current parameter, the hypothesis of \Cref{lem:crt-composition} holds.  Each composition multiplies the parameter by at most an absolute constant, so after three compositions the parameter is still $O(\Lambda)=O(\sqrt L)$.

The construction time is the base search time plus the three CRT composition times.  The latter are bounded by
\[
       L^{O(\Lambda b_0)}\poly(b)=L^{O(b\sqrt L)}\poly(b),
\]
and dominate $b^{o(1)}$ for large $b$.  This gives the desired uniform family after increasing $K_1$.
\end{proof}

\subsection{Hardness of Low-Dimensional \texorpdfstring{$\Z$}{Z}-OV}\label{sec:hardness}

\begin{proposition}[Dimensionality Reduction to $\Z$-OV]\label{prop:to-zov}
For all sufficiently large $L$ and sufficiently large $b$, let $2\le L\le d$, $b=\lceil d/L\rceil$, and assume $L\le\log\log\log b$.  Boolean $OV_{n,d}$ reduces in time
\[
      n\cdot 2^{O(b\sqrt L\log L)}+2^{O(b\sqrt L\log L)}\poly(b)
\]
to
\[
      2^{O(b\sqrt L\log L)}
\]
instances of $\Z$-OV on two sets of $n$ vectors in dimension $O(L)$, with entry bit-length $O(b\sqrt L\log L)$.
\end{proposition}

\begin{proof}
Pad each Boolean vector with zeros so that its dimension is exactly $bL$, and split it into $L$ blocks of length $b$.  Use the uniform local maps $\alpha,\beta$ and value set $V$ supplied by \Cref{lem:two-param-uniformization}.  Let $\Phi_A,\Phi_B$ be their blockwise concatenations.  For every $v\in V$, create a $\Z$-OV instance by replacing
\[
      a\in A \quad\text{with}\quad (\Phi_A(a),-v),
      \qquad
      y\in B \quad\text{with}\quad (\Phi_B(y),1).
\]
Then
\[
      (\Phi_A(a),-v)\cdot(\Phi_B(y),1)
      =\Phi_A(a)\cdot\Phi_B(y)-v.
\]
Hence the original Boolean instance has an orthogonal pair iff at least one of the produced $\Z$-OV instances has a zero inner-product pair.  The number of instances, dimension, bit-length, and reduction time follow from \Cref{lem:two-param-uniformization}.  Arithmetic on $O(b\sqrt L\log L)$-bit integers contributes only polynomial factors in that bit-length.
\end{proof}

Finally, we prove the main theorem. 

\begin{theorem}[Main Theorem]
    \label{thm:main}
    Assume OVH or SETH.  
    Let $D(n)=\omega(1)$ be a constructible dimension function. 
    Then, for every fixed $\delta>0$, $\Z$-OV on two $n$-point sets in dimension at most $D(n)$, with $O(\log n)$-bit entries, has no $O(n^{2-\delta})$-time algorithm.  
\end{theorem}

We remark that hardness of exact $\Z$-Max-IP, Furthest Pair and Bichromatic Closest Pair follows from reductions from $\Z$-OV, based on reductions from \cite{DBLP:conf/soda/Williams18, DBLP:journals/toc/Chen20}. See~\cref{app:reductions} for details. 

\begin{proof}[Proof of \Cref{thm:main}]
Fix $\delta>0$.  
It is enough to consider $0<\delta<1$; otherwise replace $\delta$ by $1/2$.  
Suppose, toward contradiction, that $\Z$-OV in dimension at most $D(n)$ with $O(\log n)$-bit entries has an $O(n^{2-\delta})$-time algorithm.  
Let $d=\ceil{c\log n}$, where $c$ is the OVH constant for $\delta/4$.

All constants and finite lower thresholds required by the reduction are fixed independent of $n$.
Let $L_0$ and $b_0$ dominate the finitely many lower bounds on $L$ and $b$ that occur in those statements.

Choose
\[
      L=L(n):=\left\lfloor
      \min\left\{\frac{D(n)}{C_9},\frac{\log\log\log\log n}{C_9}\right\}
      \right\rfloor
\]
for a sufficiently large constant $C_9$.  By \Cref{def:constructible}, this value of $L$, and hence $b=\ceil{d/L}$, can be computed in $n^{o(1)}$ time; this cost is included in the reduction time below.  Since $D(n)=\omega(1)$, for all sufficiently large $n$ we have $L\ge\max\{2,L_0\}$ and $L=\omega(1)$.  With $b=\lceil d/L\rceil$, we have
\[
      b\ge \frac{c\log n}{L},
\]
so
\[
      \log b=\log\log n-\log L+O(1)=(1+o(1))\log\log n
\]
because the second bound term in the minimization gives $\log L=o(\log\log n)$.  Hence
\[
      \log\log\log b=(1+o(1))\log\log\log\log n.
\]
In particular $b\to\infty$, so $b\ge b_0$ for all sufficiently large $n$.
Choosing $C_9$ large enough, this implies
\[
      L\le\log\log\log b
\]
for large $n$, so the uniform reduction applies.  The first term in the minimization ensures that the final $\Z$-OV dimension $O(L)$ is at most $D(n)$ for sufficiently large $C_9$; if the assumed algorithm is stated for exactly dimension $D(n)$ rather than at most $D(n)$, pad all vectors with zero coordinates.

By \Cref{prop:to-zov}, the number of produced $\Z$-OV instances is
\[
      2^{O(b\sqrt L\log L)}
      =\exp\left(O\left(\frac{d}{\sqrt L}\log L\right)\right)
      =n^{o(1)}.
\]
Indeed, $d=O(\log n)$ and $L=\omega(1)$, so
\[
      O\left(\frac{d}{\sqrt L}\log L\right)
      =O\left(\log n\cdot\frac{\log L}{\sqrt L}\right)
      =o(\log n).
\]
Their dimensions are $O(L)$ and their bit-length is
\[
      O(b\sqrt L\log L)
      =O\left(\frac{d}{\sqrt L}\log L\right)=o(\log n),
\]
so in particular it is $O(\log n)$.  The reduction time is
\[
      n\cdot n^{o(1)}+n^{o(1)}\poly(b)=n^{1+o(1)},
\]
because $b=\lceil d/L\rceil=O(\log n)$ and therefore every fixed polynomial in $b$ is $n^{o(1)}$.

Running the assumed $\Z$-OV algorithm on all instances takes
\[
      n^{o(1)}\cdot n^{2-\delta}=n^{2-\delta+o(1)}.
\]
Together with the $n^{1+o(1)}$ reduction time, this is at most $n^{2-\delta/2}$ for sufficiently large $n$, contradicting OVH for the choice of $c = c(\delta/4)$.
\end{proof}

\section{Acknowledgment}

Y. X. is grateful to Lijie Chen for discussing the Maximum Inner Product problem with him during their time at MIT, and for sharing the key insight that his proof is bottlenecked by the density of primes.

\bibliographystyle{alpha}
\bibliography{ref}

\newcommand{\etalchar}[1]{$^{#1}$}
\begin{thebibliography}{XWL{\etalchar{+}}11}

\bibitem[ABB{\etalchar{+}}25]{achim2025aristotleimolevelautomatedtheorem}
Tudor Achim, Alex Best, Alberto Bietti, Kevin Der, Mathïs Fédérico, Sergei Gukov, Daniel Halpern-Leistner, Kirsten Henningsgard, Yury Kudryashov, Alexander Meiburg, Martin Michelsen, Riley Patterson, Eric Rodriguez, Laura Scharff, Vikram Shanker, Vladmir Sicca, Hari Sowrirajan, Aidan Swope, Matyas Tamas, Vlad Tenev, Jonathan Thomm, Harold Williams, and Lawrence Wu.
\newblock Aristotle: Imo-level automated theorem proving, 2025.

\bibitem[ABV24]{DBLP:conf/mfcs/AndrejevsBV24}
Vladimirs Andrejevs, Aleksandrs Belovs, and Jevgenijs Vihrovs.
\newblock Quantum algorithms for hopcroft's problem.
\newblock In {\em Proceedings of the 49th International Symposium on Mathematical Foundations of Computer Science (MFCS)}, pages 9:1--9:16, 2024.

\bibitem[ACSS20]{DBLP:conf/focs/AlmanCS020}
Josh Alman, Timothy Chu, Aaron Schild, and Zhao Song.
\newblock Algorithms and hardness for linear algebra on geometric graphs.
\newblock In {\em Proceedings of the 61st {IEEE} Annual Symposium on Foundations of Computer Science (FOCS)}, pages 541--552, 2020.

\bibitem[AESW90]{DBLP:conf/compgeom/AgarwalESW90}
Pankaj~K. Agarwal, Herbert Edelsbrunner, Otfried Schwarzkopf, and Emo Welzl.
\newblock Euclidean minimum spanning trees and bichromatic closest pairs.
\newblock In {\em Proceedings of the 6th Annual Symposium on Computational Geometry (SoCG)}, pages 203--210, 1990.

\bibitem[AI08]{DBLP:journals/cacm/AndoniI08}
Alexandr Andoni and Piotr Indyk.
\newblock Near-optimal hashing algorithms for approximate nearest neighbor in high dimensions.
\newblock {\em Commun. {ACM}}, 51(1):117--122, 2008.

\bibitem[AIL{\etalchar{+}}15]{DBLP:conf/nips/AndoniILRS15}
Alexandr Andoni, Piotr Indyk, Thijs Laarhoven, Ilya~P. Razenshteyn, and Ludwig Schmidt.
\newblock Practical and optimal {LSH} for angular distance.
\newblock In {\em Proceedings of the 29th Annual Conference on Neural Information Processing Systems}, pages 1225--1233, 2015.

\bibitem[AINR14]{DBLP:conf/soda/AndoniINR14}
Alexandr Andoni, Piotr Indyk, Huy~L. Nguyen, and Ilya~P. Razenshteyn.
\newblock Beyond locality-sensitive hashing.
\newblock In {\em Proceedings of the 25th Annual {ACM-SIAM} Symposium on Discrete Algorithms (SODA)}, pages 1018--1028, 2014.

\bibitem[ALS{\etalchar{+}}23]{DBLP:conf/nips/AlmanL00Z23}
Josh Alman, Jiehao Liang, Zhao Song, Ruizhe Zhang, and Danyang Zhuo.
\newblock Bypass exponential time preprocessing: Fast neural network training via weight-data correlation preprocessing.
\newblock In {\em Proceedings of the 37th Annual Conference on Neural Information Processing Systems (NeurIPS)}, 2023.

\bibitem[AM18]{atiyah2018introduction}
Michael~F Atiyah and Ian~Grant Macdonald.
\newblock {\em Introduction to commutative algebra}.
\newblock CRC press, 2018.

\bibitem[APRS16]{DBLP:conf/pods/AhlePR016}
Thomas~Dybdahl Ahle, Rasmus Pagh, Ilya~P. Razenshteyn, and Francesco Silvestri.
\newblock On the complexity of inner product similarity join.
\newblock In {\em Proceedings of the 35th {ACM} {SIGMOD-SIGACT-SIGAI} Symposium on Principles of Database Systems (PODS)}, pages 151--164, 2016.

\bibitem[AR15]{DBLP:conf/stoc/AndoniR15}
Alexandr Andoni and Ilya~P. Razenshteyn.
\newblock Optimal data-dependent hashing for approximate near neighbors.
\newblock In {\em Proceedings of the 47th Annual {ACM} on Symposium on Theory of Computing (STOC)}, pages 793--801, 2015.

\bibitem[AW15]{DBLP:conf/focs/AlmanW15}
Josh Alman and Ryan Williams.
\newblock Probabilistic polynomials and hamming nearest neighbors.
\newblock In {\em Proceedings of the {IEEE} 56th Annual Symposium on Foundations of Computer Science (FOCS)}, pages 136--150, 2015.

\bibitem[BS76]{DBLP:conf/stoc/BentleyS76}
Jon~Louis Bentley and Michael~Ian Shamos.
\newblock Divide-and-conquer in multidimensional space.
\newblock In {\em Proceedings of the 8th Annual {ACM} Symposium on Theory of Computing (STOC)}, pages 220--230, 1976.

\bibitem[Cha93]{DBLP:journals/dcg/Chazelle93a}
Bernard Chazelle.
\newblock Cutting hyperplanes for divide-and-conquer.
\newblock {\em Discret. Comput. Geom.}, 9:145--158, 1993.

\bibitem[Che20]{DBLP:journals/toc/Chen20}
Lijie Chen.
\newblock On the hardness of approximate and exact (bichromatic) maximum inner product.
\newblock {\em Theory Comput.}, 16:1--50, 2020.

\bibitem[Chr17]{DBLP:conf/soda/Christiani17}
Tobias Christiani.
\newblock A framework for similarity search with space-time tradeoffs using locality-sensitive filtering.
\newblock In {\em Proceedings of the 28th Annual {ACM-SIAM} Symposium on Discrete Algorithms (SODA)}, pages 31--46, 2017.

\bibitem[CLS{\etalchar{+}}18]{chen2018fast}
Shuangmin Chen, Taijun Liu, Zhenyu Shu, Shiqing Xin, Ying He, and Changhe Tu.
\newblock Fast and robust shape diameter function.
\newblock {\em Plos one}, 13(1):e0190666, 2018.

\bibitem[CMTV00]{DBLP:conf/sigmod/CorralMTV00}
Antonio Corral, Yannis Manolopoulos, Yannis Theodoridis, and Michael Vassilakopoulos.
\newblock Closest pair queries in spatial databases.
\newblock In {\em Proceedings of the 2000 {ACM} {SIGMOD} International Conference on Management of Data}, pages 189--200, 2000.

\bibitem[CMTV04]{DBLP:journals/dke/CorralMTV04}
Antonio Corral, Yannis Manolopoulos, Yannis Theodoridis, and Michael Vassilakopoulos.
\newblock Algorithms for processing k-closest-pair queries in spatial databases.
\newblock {\em Data Knowl. Eng.}, 49(1):67--104, 2004.

\bibitem[CP17]{DBLP:conf/stoc/ChristianiP17}
Tobias Christiani and Rasmus Pagh.
\newblock Set similarity search beyond minhash.
\newblock In {\em Proceedings of the 49th Annual {ACM} {SIGACT} Symposium on Theory of Computing (STOC)}, pages 1094--1107, 2017.

\bibitem[CZ24]{DBLP:journals/talg/ChanZ24}
Timothy~M. Chan and Da~Wei Zheng.
\newblock Hopcroft's problem, log* shaving, two-dimensional fractional cascading, and decision trees.
\newblock {\em {ACM} Trans. Algorithms}, 20(3):24, 2024.

\bibitem[DF04]{dummit2004abstract}
David~S. Dummit and Richard~M. Foote.
\newblock {\em Abstract algebra}.
\newblock John Wiley \& Sons, 3rd edition, 2004.

\bibitem[DHKP97]{DBLP:journals/jal/DietzfelbingerHKP97}
Martin Dietzfelbinger, Torben Hagerup, Jyrki Katajainen, and Martti Penttonen.
\newblock A reliable randomized algorithm for the closest-pair problem.
\newblock {\em J. Algorithms}, 25(1):19--51, 1997.

\bibitem[Epp00]{DBLP:journals/jea/Eppstein00}
David Eppstein.
\newblock Fast hierarchical clustering and other applications of dynamic closest pairs.
\newblock {\em {ACM} J. Exp. Algorithmics}, 5:1, 2000.

\bibitem[Erd46]{erdos1946sets}
Paul Erd\H{o}s.
\newblock On sets of distances of n points.
\newblock {\em The American Mathematical Monthly}, 53(5):248--250, 1946.

\bibitem[Eri95]{DBLP:conf/cccg/Erickson95}
Jeff Erickson.
\newblock On the relative complexities of some geometric problems.
\newblock In {\em Proceedings of the 7th Canadian Conference on Computational Geometry (CCCG)}, pages 85--90, 1995.

\bibitem[GCI{\etalchar{+}}16]{DBLP:conf/adbis/Garcia-GarciaCI16}
Francisco Garc{\'{\i}}a{-}Garc{\'{\i}}a, Antonio Corral, Luis Iribarne, Michael Vassilakopoulos, and Yannis Manolopoulos.
\newblock Enhancing spatialhadoop with closest pair queries.
\newblock In {\em Proceedings of the 20th East European Conference on Advances in Databases and Information Systems (ADBIS)}, pages 212--225, 2016.

\bibitem[GHS{\etalchar{+}}26]{gupta2026subquadratic}
Shreya Gupta, Boyang Huang, Barna Saha, Yinzhan Xu, and Christopher Ye.
\newblock Subquadratic algorithms and hardness for attention with any temperature.
\newblock In {\em The Fourteenth International Conference on Learning Representations (ICLR)}, 2026.

\bibitem[GJK88]{DBLP:journals/trob/GilbertJK88}
Elmer~G. Gilbert, Daniel~W. Johnson, and S.~Sathiya Keerthi.
\newblock A fast procedure for computing the distance between complex objects in three-dimensional space.
\newblock {\em {IEEE} J. Robotics Autom.}, 4(2):193--203, 1988.

\bibitem[Gon85]{gonzalez1985clustering}
Teofilo~F. Gonzalez.
\newblock Clustering to minimize the maximum intercluster distance.
\newblock {\em Theoretical computer science}, 38:293--306, 1985.

\bibitem[GS64]{golod1964class}
Evgeniy~S. Golod and Igor~R. Shafarevich.
\newblock On the class field tower.
\newblock {\em Izvestiya Rossiiskoi Akademii Nauk. Seriya Matematicheskaya}, 28(2):261--272, 1964.

\bibitem[GS65]{GolodShafarevich1965}
Evgeniy~S. Golod and Igor~R. Shafarevich.
\newblock On class field towers.
\newblock In {\em Fourteen Papers on Logic, Algebra, Complex Variables and Topology}, volume~48 of {\em American Mathematical Society Translations, Series 2}, pages 91--102. American Mathematical Society, 1965.
\newblock English translation of the 1964 Russian paper.

\bibitem[Har01]{DBLP:conf/compgeom/Har-Peled01}
Sariel Har{-}Peled.
\newblock A practical approach for computing the diameter of a point set.
\newblock In {\em Proceedings of the 17th Annual Symposium on Computational Geometry (SoCG)}, pages 177--186, 2001.

\bibitem[HIM12]{DBLP:journals/toc/Har-PeledIM12}
Sariel Har{-}Peled, Piotr Indyk, and Rajeev Motwani.
\newblock Approximate nearest neighbor: Towards removing the curse of dimensionality.
\newblock {\em Theory Comput.}, 8(1):321--350, 2012.

\bibitem[KKK18]{DBLP:journals/talg/KarppaKK18}
Matti Karppa, Petteri Kaski, and Jukka Kohonen.
\newblock A faster subquadratic algorithm for finding outlier correlations.
\newblock {\em {ACM} Trans. Algorithms}, 14(3):31:1--31:26, 2018.

\bibitem[KM95]{DBLP:journals/iandc/KhullerM95}
Samir Khuller and Yossi Matias.
\newblock A simple randomized sieve algorithm for the closest-pair problem.
\newblock {\em Inf. Comput.}, 118(1):34--37, 1995.

\bibitem[LC91]{lin1991fast}
Ming~C. Lin and John~F. Canny.
\newblock A fast algorithm for incremental distance calculation.
\newblock In {\em Proceedings of the 1991 IEEE International Conference on Robotics and Automation}, pages 1008--1014, 1991.

\bibitem[Mar18]{marcus1977number}
Daniel~A. Marcus.
\newblock {\em Number fields}.
\newblock Springer, 2nd edition, 2018.

\bibitem[Mat92]{DBLP:journals/dcg/Matousek92}
Jir{\'{\i}} Matousek.
\newblock Efficient partition trees.
\newblock {\em Discret. Comput. Geom.}, 8:315--334, 1992.

\bibitem[Mat93]{DBLP:journals/dcg/Matousek93}
Jir{\'{\i}} Matousek.
\newblock Range searching with efficient hiearchical cutting.
\newblock {\em Discret. Comput. Geom.}, 10:157--182, 1993.

\bibitem[May15]{mayer2015new}
Daniel~C. Mayer.
\newblock New number fields with known p-class tower.
\newblock {\em Tatra Mountains Math. Publ}, 64:21--57, 2015.

\bibitem[MB02]{DBLP:journals/ijcga/MalandainB02}
Gr{\'{e}}goire Malandain and Jean{-}Daniel Boissonnat.
\newblock Computing the diameter of a point set.
\newblock {\em Int. J. Comput. Geom. Appl.}, 12(6):489--510, 2002.

\bibitem[Mil03]{milne2003fields}
James~S. Milne.
\newblock Fields and galois theory.
\newblock {\em Courses Notes, Version}, 4, 2003.

\bibitem[Mil20]{milne2020algebraic}
James~S. Milne.
\newblock Algebraic number theory, 2020.

\bibitem[MMC21]{DBLP:conf/pci/MavrommatisMC21}
George Mavrommatis, Panagiotis Moutafis, and Antonio Corral.
\newblock Enhancing the slicenbound algorithm for the closest-pairs query with binary space partitioning.
\newblock In {\em Proceedings of the 25th Pan-Hellenic Conference on Informatics (PCI)}, pages 107--112. {ACM}, 2021.

\bibitem[Neu13]{neukirch2013algebraic}
J{\"u}rgen Neukirch.
\newblock {\em Algebraic number theory}, volume 322.
\newblock Springer Science \& Business Media, 2013.

\bibitem[NS15]{DBLP:conf/icml/NeyshaburS15}
Behnam Neyshabur and Nathan Srebro.
\newblock On symmetric and asymmetric lshs for inner product search.
\newblock In {\em Proceedings of the 32nd International Conference on Machine Learning (ICML)}, pages 1926--1934, 2015.

\bibitem[NTM01]{DBLP:conf/vldb/ManolopoulosT01}
Alexandros Nanopoulos, Yannis Theodoridis, and Yannis Manolopoulos.
\newblock C\({}^{\mbox{2}}\)p: Clustering based on closest pairs.
\newblock In {\em Proceedings of 27th International Conference on Very Large Data Bases (VLDB)}, pages 331--340, 2001.

\bibitem[Ope26]{UnitDistance}
OpenAI.
\newblock Planar point sets with many unit distances, 2026.
\newblock Available at \url{https://cdn.openai.com/pdf/74c24085-19b0-4534-9c90-465b8e29ad73/unit-distance-proof.pdf}.

\bibitem[RG12]{DBLP:conf/kdd/RamG12}
Parikshit Ram and Alexander~G. Gray.
\newblock Maximum inner-product search using cone trees.
\newblock In {\em Proceedings of the 18th {ACM} {SIGKDD} International Conference on Knowledge Discovery and Data Mining (KDD)}, pages 931--939, 2012.

\bibitem[RLMK05]{DBLP:journals/jcise/RedonLMK05}
Stephane Redon, Ming~C. Lin, Dinesh Manocha, and Young~J. Kim.
\newblock Fast continuous collision detection for articulated models.
\newblock {\em J. Comput. Inf. Sci. Eng.}, 5(2):126--137, 2005.

\bibitem[RR07]{DBLP:conf/nips/RahimiR07}
Ali Rahimi and Benjamin Recht.
\newblock Random features for large-scale kernel machines.
\newblock In {\em Proceedings of the 21st Annual Conference on Neural Information Processing Systems}, pages 1177--1184, 2007.

\bibitem[{\v{S}}af63]{vsafarevivc1963extensions}
Igor~R. {\v{S}}afarevi{\v{c}}.
\newblock Extensions {\`a} points de ramification donn{\'e}s (en russe).
\newblock {\em Publications Math{\'e}matiques de l'IH{\'E}S}, 18:71--92, 1963.

\bibitem[SH75]{DBLP:conf/focs/ShamosH75}
Michael~Ian Shamos and Dan Hoey.
\newblock Closest-point problems.
\newblock In {\em Proceedings of the 16th Annual Symposium on Foundations of Computer Science (FOCS)}, pages 151--162, 1975.

\bibitem[Sha66]{shafarevich1966extensions}
Igor~R. Shafarevich.
\newblock Extensions with given points of ramification.
\newblock {\em Amer. Math. Soc. Translation, Ser}, 2(59):128--149, 1966.

\bibitem[SL14]{DBLP:conf/nips/Shrivastava014}
Anshumali Shrivastava and Ping Li.
\newblock Asymmetric {LSH} {(ALSH)} for sublinear time maximum inner product search {(MIPS)}.
\newblock In {\em Proceedings of the 28th Annual Conference on Neural Information Processing Systems}, pages 2321--2329, 2014.

\bibitem[SL15]{DBLP:conf/www/Shrivastava015}
Anshumali Shrivastava and Ping Li.
\newblock Asymmetric minwise hashing for indexing binary inner products and set containment.
\newblock In {\em Proceedings of the 24th International Conference on World Wide Web (WWW)}, pages 981--991, 2015.

\bibitem[SSC08]{DBLP:journals/vc/ShapiraSC08}
Lior Shapira, Ariel Shamir, and Daniel Cohen{-}Or.
\newblock Consistent mesh partitioning and skeletonisation using the shape diameter function.
\newblock {\em Vis. Comput.}, 24(4):249--259, 2008.

\bibitem[TG17]{DBLP:journals/tods/TeflioudiG17}
Christina Teflioudi and Rainer Gemulla.
\newblock Exact and approximate maximum inner product search with {LEMP}.
\newblock {\em {ACM} Trans. Database Syst.}, 42(1):5:1--5:49, 2017.

\bibitem[{The}26]{erdos_unit_distance_formalization_2026}
{The Erd\H{o}s unit-distance formalization contributors}.
\newblock {Lean 4 formalization of the disproof of Erd\H{o}s's planar unit-distance conjecture}, 2026.
\newblock \url{https://github.com/logical-intelligence/erdos-unit-distance}.

\bibitem[Tou83]{toussaint1983solving}
Godfried~T. Toussaint.
\newblock Solving geometric problems with the rotating calipers.
\newblock In {\em Proc. IEEE Melecon}, volume~83, page A10, 1983.

\bibitem[Tsc26]{tschebotareff1926bestimmung}
Nikolaj Tschebotareff.
\newblock Die bestimmung der dichtigkeit einer menge von primzahlen, welche zu einer gegebenen substitutionsklasse geh{\"o}ren.
\newblock {\em Mathematische Annalen}, 95(1):191--228, 1926.

\bibitem[Val15]{DBLP:journals/jacm/Valiant15}
Gregory Valiant.
\newblock Finding correlations in subquadratic time, with applications to learning parities and the closest pair problem.
\newblock {\em J. {ACM}}, 62(2):13:1--13:45, 2015.

\bibitem[Wil05]{DBLP:journals/tcs/Williams05}
Ryan Williams.
\newblock A new algorithm for optimal 2-constraint satisfaction and its implications.
\newblock {\em Theor. Comput. Sci.}, 348(2-3):357--365, 2005.

\bibitem[Wil18]{DBLP:conf/soda/Williams18}
Ryan Williams.
\newblock On the difference between closest, furthest, and orthogonal pairs: Nearly-linear vs barely-subquadratic complexity.
\newblock In {\em Proceedings of the 29th Annual {ACM-SIAM} Symposium on Discrete Algorithms (SODA)}, pages 1207--1215, 2018.

\bibitem[Wil20]{wilkes2020part}
Gareth Wilkes.
\newblock Part iii profinite groups.
\newblock Lecture notes, University of Cambridge, 2020.
\newblock Available at \url{https://www.dpmms.cam.ac.uk/~grw46/LectureNotes.pdf}.

\bibitem[XWL{\etalchar{+}}11]{xiao2011efficient}
Chuan Xiao, Wei Wang, Xuemin Lin, Jeffrey~Xu Yu, and Guoren Wang.
\newblock Efficient similarity joins for near-duplicate detection.
\newblock {\em ACM Transactions on Database Systems (TODS)}, 36(3):1--41, 2011.

\bibitem[Yao82]{yao1982constructing}
Andrew Chi-Chih Yao.
\newblock On constructing minimum spanning trees in k-dimensional spaces and related problems.
\newblock {\em SIAM Journal on Computing}, 11(4):721--736, 1982.

\end{thebibliography}

\appendix 
\newpage 

\section{Tools from Number Theory}
\label{sec:arithmetic}

The goal of this section is to prove \Cref{prop:codebook}.
This section begins with \Cref{thm:ud-external}, from which the remaining proof of \Cref{prop:codebook} will be self-contained.

\subsection{Further Group Theory and Number Theory Preliminaries}

We begin with the necessary preliminaries (beyond those given in \Cref{sec:number-theory-prelims}).
We use $H_2(x) = x \log(1/x) + (1 - x) \log(1/(1 - x))$ to denote the binary entropy function with natural logarithms.
Let $G$ be a group.
A subgroup $H \triangleleft G$ is \emph{normal} if $g H g^{-1} = H$ for all $g \in G$, i.e. it is invariant under conjugation.
The \emph{order} of a group is its size $|G|$.
A normal subgroup $H \triangleleft G$ has \emph{index} $[G : H] = |G/H|$.

\begin{definition}[Normal Extension (Chapter 13 of \cite{dummit2004abstract})]
    \label{def:normal}
    The extension field $K/F$ is called a splitting field for the polynomial $f(x)$ if $f(x)$ factors completely into linear factors (or splits completely) in $K[x]$ (i.e. splits into degree-$1$ polynomials with coefficients in $K$) and $f(x)$ does not factor completely into linear factors over any proper subfield of $K$ containing $F$.
    An extension $K/F$ is normal if $K$ is the splitting field for some collection of polynomials.
    The normal closure of an extension $K/F$ is the smallest normal extension $K/F \subseteq N/F$.
\end{definition}

\begin{definition}[Separable (Chapter 13 of \cite{dummit2004abstract})]
    \label{def:separable}
    A polynomial over $F$ is called separable if it has no multiple roots (i.e., all its roots are distinct).
    The field $K$ is said to be separable (or separably algebraic) over $F$ if
    every element of $K$ is the root of a separable polynomial over $F$ (equivalently, the
    minimal polynomial over $F$ of every element of $K$ is separable). 
\end{definition}

\begin{definition}[Galois Extension ({\cite[Chapter 14]{dummit2004abstract}})]
    \label{def:galois}
    A finite extension $K/F$ is Galois if $|\Aut(K/F)| = [K:F]$.
    In this case, the Galois group of $K/F$ is $\Gal(K/F) = \Aut(K/F)$.
\end{definition}

We note that Galois extensions are normal and separable. 
Furthermore, since any extension of $\Q$ is separable, such an extension is Galois iff it is normal.

\begin{definition}[Fixed Field \cite{dummit2004abstract}]
    \label{def:fixed-field}
    If $H$ is a subgroup of $\Aut(K/F)$, the subfield of $K$ fixed by all the elements of $H$ is called the fixed field of $H$, denoted $K^{H}$.
\end{definition}

Next, we define $p$-groups and pro-$p$ groups.

\begin{definition}
    \label{def:p-group}
    A finite group $G$ is a $p$-group if every element has order that is a power of $p$.
\end{definition}

\begin{definition}[Inverse Limits \cite{wilkes2020part}]
    \label{def:inverse-limit}
    A partially ordered set $(J, \preceq)$ is an inverse system if for any $i,j \in J$ there is some $k \in J$ such that $k \preceq i$ and $k \preceq j$.
    
    An inverse system of groups $(J, \preceq)$ consists of: (1) group $G_j$ for each $j \in J$ and, whenever $i \preceq j$, we have some homomorphism $\phi_{ij}: G_{i} \rightarrow G_{j}$ such that $\phi_{ii} = \id_{G_{i}}$ and $\phi_{ik} = \phi_{jk} \circ \phi_{ij}$ when $i \preceq j \preceq k$.
    
    The inverse limit of an inverse system of groups (more informally, the inverse limit of the groups $G_{j}$) is the group, denoted $\underset{\leftarrow}{\lim} G_{j}$, of sequences $(g_{j})_{j \in J}$ such that $g_{j} \in G_{j}$ for all $j$ and $\phi_{ij}(g_{i}) = g_{j}$ for all $i \preceq j$.
\end{definition}

\begin{definition}[Pro-$p$ Groups \cite{wilkes2020part}]
    \label{def:pro-p-group}
    Let $p$ be a prime. A pro-$p$ group is an inverse limit of an inverse system of finite $p$-groups.
\end{definition}

\begin{definition}[Maximal unramified pro-p extension (Definition A.3 of \cite{UnitDistance})]
    \label{def:compositum}
    For a number field $F$, $F^{\ur,p}$ denotes the compositum of all finite everywhere-unramified Galois extensions of $F$ whose Galois groups are finite $p$-groups. 
    Its Galois group $\Gal(F^{\ur,p}/F)$ is a pro-$p$ group, and its finite quotients correspond to finite everywhere-unramified Galois $p$-group extensions of $F$.
\end{definition}

Pro-$p$ groups are equipped with the following topology.

\begin{definition}[\cite{wilkes2020part}]
    Let $(G_j)_{j \in J}$ be an inverse system of finite groups. 
    Endow each $G_j$ with the discrete topology (every subset is open and closed), and give $\prod G_j$ the product topology (open sets are the full group $G_j$ in all but finitely many coordinates).
    The topology on $\underset{\leftarrow}{\lim}\,G_j \subseteq \prod G_j$ is the subspace topology.
\end{definition}

\begin{definition}[Frattini Subgroup (Definition 4.1.1 of \cite{wilkes2020part})] 
    \label{def:frattini}
    The Frattini subgroup of group $G$, denoted $\Phi(G)$, is the intersection of all maximal proper subgroups of $G$.
\end{definition}

We define generator rank and relation rank of pro-$p$ groups.

\begin{definition}[Generator rank \cite{wilkes2020part}]
    Let $G$ be a topological group and let $S$ be a subset of $G$. 
    We say $S$ is a (topological) generating set for $G$ if the subgroup $\langle S \rangle$ generated by $S$ is a dense subgroup of $G$.\footnote{A subset is dense if it intersects every non-empty open set. Equivalently, its closure is the whole set.}
    The group $G$ is (topologically) finitely generated if it has a finite topological generating set.
    The generator rank of $G$, denoted $d(G)$, is the minimum size of a (topological) generating set.
\end{definition}

For a topologically finitely generated pro-$p$ group $G$, we have $H_1(G, \F_{p}) := G/\Phi(G) \cong \F_{p}^{d(G)}$ (see e.g. Proposition 4.1.13 of \cite{wilkes2020part}, and \Cref{prop:pro-p-dim}).

\begin{definition}[Relation rank (Section 5.4.2 of \cite{wilkes2020part})]
    \label{def:relation-rank}
    Let $G$ be a topologically finitely generated pro-$p$ group, and let $d(G)$ denote its minimal number of topological generators. 
    A finite pro-$p$ presentation of $G$ is an isomorphism
    \[
        G \cong F/\langle\!\langle R\rangle\!\rangle \text{,}
    \]
    where $F$ is a finitely generated free pro-$p$ group and $\langle\!\langle R\rangle\!\rangle$ denotes the closed normal subgroup of $F$ topologically normally generated by $R$. 
    The relation rank of $G$, denoted $r(G)$, is the minimum cardinality of such a set $R$, taken over pro-$p$ presentations of $G$ in which $F$ is freely generated by $d(G)$ elements.
\end{definition}

For a topologically finitely generated pro-$p$ group $G$, one has
\[
    r(G) = \dim_{\F_p} H^2(G,\F_p) \text{,}
\]
where $H^2$ denotes the second cohomology group (see e.g. Definition 5.1.16 and Theorem 5.4.28 of \cite{wilkes2020part}).
In particular, $r(G)$ is intrinsic to $G$ and does not depend on the choice of a minimal topological generating set.

\begin{definition}[Trace (see e.g. Corollary 2.20 of \cite{milne2020algebraic})]
    \label{def:trace-norm}
    Let $L/K$ be a separable extension so that there are $n = [L:K]$ distinct embeddings of $L$ fixing $K$, denoted $\set{\sigma_{i}}$.
    Then, define for $x \in L$, $\Tr_{L/K}(x) = \sum_{i = 1}^{n} \sigma_{i}(x)$.
\end{definition}

\begin{definition}[Relative Discriminant (\cite{milne2020algebraic})]
    \label{def:relative-discriminant}
    Let $L/K$ be a separable extension of degree $n = [L:K]$ and $x_1, \dots, x_{n} \in L$.
    Then their discriminant is
    \begin{equation*}
        \disc(x_1, \dots, x_n) = \det(\Tr_{L/K}(x_i x_j)) \text{.} 
    \end{equation*}
    The relative discriminant $\frakd_{L/K}$ is the ideal of $\calO_{K}$ generated by $\disc(x_1, \dots, x_n)$ for all $x_1, \dots, x_n \in \calO_{L}$.
\end{definition}

\begin{definition}[Absolute Discriminant (\cite{milne2020algebraic})]
    \label{def:absolute-discriminant}
    The absolute discriminant of a number field $K/\Q$ is defined
    \begin{equation*}
        |\Delta_{K}| = |\disc(x_1, \dots, x_n)|
    \end{equation*}
    where $\set{x_i}$ is (any) $\Z$-basis of $\calO_{K}$.
\end{definition}

We use the following properties of the discriminant given in \cite{UnitDistance}.

\begin{definition}[Discriminant (Definition A.6 of \cite{UnitDistance})]
    \label{def:discriminant}
    The root discriminant of a number field $K/\Q$ is $\rd(K) = |\Delta_{K}|^{1/[K:\Q]}$.
    Finite unramified extensions preserve root discriminant.

    For an extension $M/F$, the relative discriminant $\frakd_{M/F}$ is an ideal of $\calO_F$, and the tower formula is 
    \begin{equation*}
        |\Delta_M|= |\Delta_F|^{[M:F]} N_{F/\Q}(\frakd_{M/F}).
    \end{equation*}
    Thus $\frakd_{M/F} = \calO_{F}$ is equivalent to no finite ramification in $M/F$. 
    Here, $N_{F/\Q}(\fraka) = [\calO_{F} : \fraka]$ denotes the absolute norm of an ideal $\fraka$.
\end{definition}

Next, we define the class number.

\begin{definition}[Fractional Ideal \cite{neukirch2013algebraic}]
    \label{def:fractional=ideal}
    A fractional ideal of $K$ is a finitely generated $\calO_{K}$-submodule $\fraka \neq 0$ of $K$.
    The ideal group of $K$, denoted $J_K$ is the group of fractional ideals of $K$, with identity $\calO_{K}$ and inverse $\fraka^{-1} = \set{x \in K \mid x \fraka \subseteq \calO_{K}}$.
    The fractional principal ideals of $K$, denoted $P_{K}$, form a subgroup of $J_K$.
\end{definition}

\begin{definition}[Class Number \cite{neukirch2013algebraic}]
    \label{def:class-number}
    The class group of $K$ is the group $J_K/P_K$. 
    The class number of $K$, denoted $\classnum(K)$ is $\classnum(K) = |J_K/P_K|$. 
\end{definition}

Next, we define the Frobenius representative.

\begin{definition}[Frobenius Representative \cite{marcus1977number}]
    \label{def:frobenius}
    Let $L/K$ be a Galois extension.
    Let $\frakp$ be a prime ideal of $\calO_{K}$ that is unramified in $L/K$, and $\frakP$ a prime lying above $\frakp$ in $\calO_{L}$.
    The Frobenius automorphism of $\frakP$, denoted $\phi_{\frakP}$ is the unique $\phi \in \Gal(L/K)$ such that $\phi_{\frakP}(\alpha) \equiv \alpha^{|\calO_{K}/\frakp|} \pmod{\frakP}$ for all $\alpha \in \calO_{L}$.

\end{definition}

We state the necessary preliminaries for decomposition groups.

\begin{definition}[Decomposition Group \cite{marcus1977number,milne2020algebraic}]
    Let $L/K$ be an extension and $\frakp$ a prime ideal in $\calO_{K}$.
    For a prime $\frakP$ above $\frakp$ in $\calO_{L}$, the decomposition group $D(\frakP \mid \frakp)$ is the set of automorphisms in $\Gal(L/K)$ satisfying $\sigma(\frakP) = \frakP$.

    For an infinite place $v$ of $K$ and $w$ above $v$ (i.e. agreeing with $v$ on $K$), the decomposition group of $w$, denoted $D_{w}$, is the set of automorphisms in $\Gal(L/K)$ such that $\sigma(w) = w$.
\end{definition}

\begin{definition}[Chapter 7 of \cite{milne2020algebraic}]
    \label{def:p-adic-value}
    For any prime ideal $\frakp$ of number field $K$, the $\frakp$-adic absolute value is given by $|x|_{\frakp} = (1/e)^{\val_{\frakp}(x)}$ for some constant $e > 1$.
    The standard choice is $e = \N \frakp := |\calO_{K}/\frakp|$ (Chapter 4 of \cite{milne2020algebraic}).
    Let $\Q_{p}$ denote the completion of $\Q$ with the $p$-adic absolute value.
\end{definition}

\begin{definition}[Chapter 7 of \cite{milne2020algebraic}]
    \label{def:completion}
    Let $K$ be a field and $|\cdot |$ a non-trivial absolute value.
    A sequence $\set{x_i}$ is Cauchy if for all $\eps > 0$, there exists $N$ such that $|x_i - x_j| < \eps$ for all $i, j \geq N$.
    A field $K$ is complete if every Cauchy sequence has a limit in $K$.
\end{definition}

\begin{claim}[Theorem 7.23 of \cite{milne2020algebraic}]
    \label{clm:completion-homomorphism}
    Let $K$ be a field with an absolute value $|\cdot |$. Then there exists a complete field $\hat{K}$ and a homomorphism $K \rightarrow \hat{K}$ preserving the absolute value that is universal in the following sense: every homomorphism $K \rightarrow L$ from $K$ into a complete field $L$ preserving the absolute value, extends uniquely to a homomorphism $\hat{K} \rightarrow L$.

    For a place $v$ of $K$, we let $K_v$ denote the completion of $K$ w.r.t. $v$.
    For a prime ideal $\frakp$, $K_{\frakp}$ denotes the completion of $K$ w.r.t. the $p$-adic absolute value $|\cdot|_{\frakp}$.
\end{claim}

\begin{claim}[Theorem 7.38, Corollary 7.42 of \cite{milne2020algebraic}]
    \label{clm:local-degree-formula}
    Let $K$ be a complete field w.r.t. a discrete absolute value $|\cdot|_{K}$ and $L/K$ be a finite, separable extension.
    Then $[L:K] = e f$ where $e$ is the ramification index and $f$ is the residue field degree.
\end{claim}

In particular, for a number field $K/\Q$, a rational prime $q$, and a prime ideal $\frakP$ of $\calO_{K}$ lying above $q$, the completion $K_{\frakP}$ is a finite extension of $\Q_q$ (see e.g. Remark 8.3 of \cite{milne2020algebraic}), and $[K_{\frakP} : \Q_{q}] = e(\frakP|q) f(\frakP|q)$.

\begin{claim}[Proposition 8.10 of \cite{milne2020algebraic}]
    \label{clm:decomposition-groups}
    Let $L/K$ be a finite Galois extension of number field $K$.
    Let $v$ be an infinite place of $K$ and $w$ a place of $L$ lying above $v$.
    Let $K_{v}$ denote the completion of $K$ with respect to $v$ and $L_{w}$ the completion of $L$ with respect to $w$.
    Then, $K_v, L_{w}$ are either isomorphic to $\R$ or $\C$ (see Remark 7.49 of \cite{milne2020algebraic}) and $D_{w} \cong \Gal(L_{w}/K_{v})$ .
\end{claim}

We will use the following facts.

\begin{proposition}[Proposition 3.3 of \cite{UnitDistance}]
    \label{prop:pro-p-dim}
    Let $G$ be a finitely generated pro-$p$ group. Then $d(G) = \dim_{\F_p}(G/\Phi(G))$. 
    If $g_1,\dots,g_k \in \Phi(G)$ and $N$ is their closed normal closure, then $d(G/N) = d(G)$ and $r(G/N) \leq r(G) + k$.
\end{proposition}

\begin{proposition}[Golod-Shafarevich Inequality ({\cite[Proposition 3.4]{UnitDistance}})]
    If a finite non-trivial pro-$p$ group has generator rank $d$ and relation rank $r$, then $r>d^2/4$. 
    Equivalently, a nontrivial finitely generated pro-$p$ group with $r \leq d^2/4$ is infinite.
    \label{prop:golod-shafarevich}
\end{proposition}

\begin{proposition}[Shafarevich Relation-Rank Estimate (Proposition 3.5 of \cite{UnitDistance})]
    \label{prop:shafarevich}
    Let $F$ be a totally real cubic field, so $\zeta_3 \not\in F$, and let $G = \Gal(F^{\ur, 3}/F)$ be the Galois group of its maximal everywhere-unramified
    pro-$3$ extension.
    Then $r(G) \leq d(G) + C_0$ for an absolute constant
    $C_0$.
\end{proposition}

\begin{proposition}[Proposition 3.7 of \cite{UnitDistance}]
    \label{prop:class-number}
    There is an absolute constant $C_{\class} > 0$ such that every number field $K$ satisfies $h(K) \leq \max(2, \rd(K))^{C_{\class}[K:\Q]}$. The constant $C_{\class}$ is absolute: it is independent of the field $K$, its degree, and its signature. For fields with $\rd(K) \geq 2$, this is simply $h(K) \leq \rd(K)^{O([K:\Q])} = |\Delta_K|^{O(1)}$.
\end{proposition}

\subsection{Proof of \texorpdfstring{\Cref{thm:ud-external}}{}}

We now state \Cref{thm:ud-external} and provide its proof.
Every subsequent statement in this section is proved below from \Cref{thm:ud-external}.

\begin{theorem}[Number-Theoretic Input from \cite{UnitDistance}]\label{thm:ud-external}
There are absolute constants $C_{\rm Sha},C_{\rm class},C_{\rm rd}>0$ such that, for all sufficiently large integers $\ell$, there is a totally real cyclic cubic field $F_\ell$ with the following properties.
Let
\[
      G_\ell:=\Gal(F_\ell^{\rm ur,3}/F_\ell)
\]
be the Galois group of the maximal everywhere-unramified pro-$3$ extension of $F_\ell$, and let $\Phi(G_\ell)$ be its Frattini subgroup.  Then:
\begin{enumerate}[label=(\roman*)]
    \item $d(G_\ell)\ge \ell-1$ and $r(G_\ell)\le d(G_\ell)+C_{\rm Sha}$, where $d$ and $r$ denote generator and relation ranks;
    \item $\log\operatorname{rd}(F_\ell)\le C_{\rm rd}\ell\log\ell$;
    \item if $E_\ell/F_\ell$ is the finite extension corresponding to $G_\ell/\Phi(G_\ell)$ and $N_\ell$ is the finite normal closure over $\Q$ of $E_\ell(i)$, then every rational prime $q$ splitting completely in $N_\ell$ satisfies $q\equiv1\pmod4$, splits completely in $F_\ell$, and, for every prime $v\mid q$ of $F_\ell$, has a Frobenius representative in $\Phi(G_\ell)$;
    \item if $g_1,\ldots,g_k\in\Phi(G_\ell)$ and $N$ is their closed normal closure in $G_\ell$, then $d(G_\ell/N)=d(G_\ell)$ and $r(G_\ell/N)\le r(G_\ell)+k$;
    \item a nontrivial finitely generated pro-$p$ group with generator rank $d$ and relation rank $r\le d^2/4$ is infinite;
    \item every number field $K$ satisfies
    \[
         h(K)\le \max\{2,\operatorname{rd}(K)\}^{C_{\rm class}[K:\Q]} .
    \]
\end{enumerate}
\end{theorem}

\begin{proof}
    The field $F_{\ell}$ is given by Proposition 3.8 of \cite{UnitDistance}. 
    We verify each property step by step.
    
    The first property follows from Equation (5) and Equation (7) of \cite{UnitDistance}.
    We note that the second inequality follows from the Shafarevich relation-rank estimate (\Cref{prop:shafarevich}).

    The second property follows from Equation (6) of \cite{UnitDistance}.

    We will argue the third property.
    Let $q$ be a rational prime splitting completely in $N_\ell$, the normal closure over $\Q$ of $E_\ell(i)$. Since $\Q(i)$, $F_\ell$, and $E_\ell$ are subfields of $N_\ell$, complete splitting in $N_\ell$ descends to each of these subextensions. Thus $q$ splits in $\Q(i)$, so $q\equiv1\pmod4$, and $q$ splits completely in $F_\ell$. Moreover, for every prime $v\mid q$ of $F_\ell$, the prime $v$ splits completely in $E_\ell/F_\ell$, so its Frobenius class in $\Gal(E_\ell/F_\ell)\cong G_\ell/\Phi(G_\ell)$ is trivial. Hence any Frobenius representative at $v$ in $G_\ell$ lies in $\Phi(G_\ell)$.

    We note that $G_{\ell}$ is finitely generated (as argued in \cite{UnitDistance}).
    The fourth property then follows from \Cref{prop:pro-p-dim}, since $G_{\ell}$ is a pro-$3$ group.

    The fifth property follows from \Cref{prop:golod-shafarevich}.

    The sixth property follows from \Cref{prop:class-number}.
\end{proof}

\subsection{Proof of \texorpdfstring{\Cref{prop:codebook}}{}}

The next proposition provides the modifications we need from \Cref{thm:ud-external}.

\begin{proposition}
\label{prop:ud-package}
There are absolute constants $c_0,C_H>0$ such that, for all sufficiently large integers $\ell$, $F_\ell,N_\ell$ from \Cref{thm:ud-external} supply integers and constants
\[
       T_\ell\in\Z_{\ge1},\qquad T_\ell\ge c_0\ell^2,\qquad
       H_\ell\ge1,\qquad
       \log H_\ell\le C_H\ell\log\ell,
\]
a totally real base field $L_{0,\ell}:=F_\ell$ of degree $d_\ell=3$, and the finite normal extension $N_\ell/\Q$, which contains $\Q(i)$ and the normal closure of $L_{0,\ell}$, with the following property.  For every set $\mathcal Q$ of at most $T_\ell$ rational primes splitting completely in $N_\ell$, there is an infinite pro-$3$ Galois extension
\[
       \mathcal L_{\ell,\mathcal Q}/L_{0,\ell}
\]
such that for every finite subextension $L/L_{0,\ell}$:
\begin{enumerate}[label=(\roman*)]
    \item $L/L_{0, \ell}$ is everywhere unramified and totally real;
    \item every $q\in\mathcal Q$ splits completely in $L$;
    \item with $K:=L(i)$, one has
    \[
          h(K)\le H_\ell^{[L:\Q]} .
    \]
\end{enumerate}
The constants $N_\ell,L_{0,\ell},H_\ell$ depend only on $\ell$, not on the particular primes in $\mathcal Q$.  
\end{proposition}

\begin{proof}
Fix $\ell$ large enough that \Cref{thm:ud-external} applies and that the numerical inequalities below hold.  Put
\[
      T_\ell:=\left\lfloor\frac{(\ell-1)^2}{100}\right\rfloor,
      \qquad
      c_0:=\frac1{200}.
\]
For sufficiently large $\ell$, we have $T_\ell\ge c_0\ell^2$.  
The field $N_\ell$ contains $\Q(i)$ by definition.
Since it is normal over $\Q$ and contains $F_\ell=L_{0,\ell}$, it also contains the normal closure of $L_{0,\ell}$.

Let $\mathcal Q$ be any set of at most $T_\ell$ rational primes splitting completely in $N_\ell$.  For each $q\in\mathcal Q$ and each prime $v\mid q$ of $F_\ell$, choose a Frobenius representative $\sigma_v\in G_\ell$.\footnote{Technically, Frobenius representatives are chosen with respect to some prime above $v$. Here, since the Frobenius representatives for different choices of primes above $v$ are equivalent up to conjugation (see e.g. \cite{marcus1977number}), we let $\sigma_{v}$ denote any  representative in the conjugacy class.}
By \Cref{thm:ud-external} (iii), every such $\sigma_v$ lies in $\Phi(G_\ell)$ and $q$ splits completely in $F_\ell$.  
Since $F_\ell/\Q$ is cubic and $q$ splits completely in $F_\ell$, there are exactly $3|\mathcal Q|$ such primes $v$.

Let $N(\mathcal Q)$ be the closed normal\footnote{Since Frobenius representatives are equivalent up to conjugacy class, the normal subgroup is well defined.} subgroup of $G_\ell$ generated by these $3|\mathcal Q|$ elements (i.e. $\sigma_v \in \Phi(G_{\ell})$), and set
\[
      \overline G_{\ell,\mathcal Q}:=G_\ell/N(\mathcal Q) \text{.}
\]
Write $d:=d(G_\ell)$.  By \Cref{thm:ud-external} (i) and (iv),
\[
      d(\overline G_{\ell,\mathcal Q})=d,
      \qquad
      r(\overline G_{\ell,\mathcal Q})
      \le d+C_{\rm Sha}+3|\mathcal Q|
      \le d+C_{\rm Sha}+\frac{3(\ell-1)^2}{100}.
\]
Because $d\ge\ell-1$, the last term is at most $d+C_{\rm Sha}+3d^2/100$.  For all sufficiently large $\ell$ this is strictly less than $d^2/4$.  Also $d(\overline G_{\ell,\mathcal Q})=d>0$, so the quotient is nontrivial.  Hence \Cref{thm:ud-external} (v) implies that $\overline G_{\ell,\mathcal Q}$ is infinite.

Let $\mathcal L_{\ell,\mathcal Q}$ be the fixed field of $N(\mathcal Q)$ inside $F_\ell^{\rm ur,3}$.  
By the fundamental theorem of infinite Galois theory (see e.g. Theorem 7.13 of \cite{milne2003fields}), the intermediate extension $\calL_{\ell, \calQ}$ with Galois group $\Gal(\mathcal L_{\ell,\mathcal Q}/F_\ell) \cong \overline{G}_{\ell, \calQ}$ is an infinite Galois pro-$3$ extension.
Indeed, we have argued the infinite condition above, and quotients of the pro-$3$ group $G_{\ell}$ by closed normal subgroups are pro-$3$ groups.

We now verify the remaining properties.
Let $L/L_{0, \ell}$ be a finite subextension.
We claim that any finite subextension is everywhere unramified over $F_\ell$.
We use the following standard facts (see e.g. Theorems 29 and 31 of \cite{marcus1977number}): 
\begin{enumerate}
    \item if a finite prime is unramified in $L_1/K$ and $L_2/K$, then it is also unramified in their compositum (i.e. the extension generated by $L_1/K$ and $L_2/K$),
    \item if a finite prime is unramified in an extension $M/K$, then it is unramified in any intermediate extension $K \subseteq M \subseteq L$.
\end{enumerate}
Consider a finite sub-extension $L = F_{\ell}(\alpha_1, \dots, \alpha_m)$ and a finite prime $p$.
Note each $\alpha_i$ lies in some finite everywhere-unramified Galois extension $E_i/F_\ell$ where $p$ is unramified.
Then, $L \subseteq E_1 E_2 \dots E_m$ is an intermediate extension of a finite compositum of extensions where $p$ is unramified, so that $p$ is unramified in $L$.
To argue infinite primes are unramified, we observe that each finite subextension is totally real. Indeed, if such an $L$ were not totally real, then after passing to its finite Galois closure $E/F_\ell$ inside $\mathcal L_{\ell,\mathcal Q}$, some real place of $F_\ell$ would have a complex place of $E$ above it. By \Cref{clm:decomposition-groups}, the corresponding decomposition group is isomorphic to $\Gal(\C/\R)$ and therefore has order $2$, which is impossible inside the finite quotient $\Gal(E/F_\ell)$ of the pro-$3$ group.
Since $F_{\ell}$ is totally real, so is every finite subextension.

Next, we argue that every prime $q \in \calQ$ splits completely in $L$.
For $q \in\mathcal Q$, the prime $q$ splits completely in $F_\ell$ by \Cref{thm:ud-external} (iii).  
For every prime $v \mid q$ of $F_\ell$, the chosen Frobenius representative $\sigma_v$ maps to the identity in $\overline G_{\ell,\mathcal Q}$ by construction.
First, take some finite intermediate Galois subextension $E$, so that Galois correspondence guarantees $E$ is the fixed field of some normal subgroup $U \triangleleft G_{\ell}$. 
Since $E \subseteq \calL_{\ell, \calQ}$ we have $N(\calQ) \subseteq U$ so that $\sigma_{v}$ maps to the identity in $G_{\ell}/U$.
Since an unramified prime $v$ splits completely in a normal extension iff the Frobenius representative $\sigma_v = 1$ (see e.g. Theorem 32 of \cite{marcus1977number}), $v$ splits completely in every finite Galois subextension of $\mathcal L_{\ell,\mathcal Q}/F_\ell$.
Hence also in every finite subextension (e.g. $L$) since we can choose some finite Galois sub-extension $E$ containing $L$ (taking the finite Galois closure of $L/F_\ell$ inside $\mathcal L_{\ell,\mathcal Q}$) and above we have argued that $v$ splits completely in $E$.
Since $q$ already splits completely in $F_\ell$ and every prime $v$ above $q$ splits completely in $L$, it follows that $q$ splits completely in every finite subextension $L/F_\ell$.

It remains to prove the class-number bound with a constant independent of $\mathcal Q$.  
Let $L/F_\ell$ be a finite subextension and set $K=L(i)$.  
Since $L/F_\ell$ is unramified at finite primes, $\operatorname{rd}(L)=\operatorname{rd}(F_\ell)$ (\Cref{def:discriminant}). 
Since the two $L$-embeddings of $K$ are the identity and complex conjugation and $1, i \in \calO_{K}$, we have that $\frakd_{K/L}$ contains
\begin{equation*}
    \disc(1, i) = \det\begin{pmatrix}
        \Tr(1) & \Tr(i) \\
        \Tr(i) & \Tr(-1)
    \end{pmatrix} = \det\begin{pmatrix}
        2 & 0 \\
        0 & -2
    \end{pmatrix} = - 4 \text{.}
\end{equation*}
Hence, $(4) \subseteq \frakd_{K/L}$, or equivalently, $\frakd_{K/L} \mid 4 \OO_L$. 
Therefore,
\[
      |\Delta_K|
      = |\Delta_L|^2\,\Norm_{L/\Q}(\mathfrak d_{K/L})
      \le |\Delta_L|^2\,4^{[L:\Q]}.
\]
The equality follows from \Cref{def:discriminant} and observing $[K:L] = 2$.
The inequality follows from the fact that $(4) \subseteq \frakd_{K/L}$ implies $\Norm_{L/\Q}(\frakd_{K/L}) \leq \Norm_{L/\Q}((4)) = 4^{[L:\Q]}$ since every embedding fixes $\Q$.
Taking $2[L:\Q]$-th roots gives
\[
      \operatorname{rd}(K)\le 2\operatorname{rd}(F_\ell).
\]
By \Cref{thm:ud-external} (vi),
\[
      h(K)\le \max\{2,2\operatorname{rd}(F_\ell)\}^{2C_{\rm class}[L:\Q]}.
\]
Define
\[
      H_\ell:=\max\{2,2\operatorname{rd}(F_\ell)\}^{2C_{\rm class}} .
\]
Then $h(K)\le H_\ell^{[L:\Q]}$.  This $H_\ell$ depends only on $\ell$, not on $\mathcal Q$.  Finally, by \Cref{thm:ud-external} (ii),
\[
      \log H_\ell
      \le 2C_{\rm class}\log\max\{2,2\operatorname{rd}(F_\ell)\}
      \le C_H\ell\log\ell
\]
for a sufficiently large absolute constant $C_H$ and all sufficiently large $\ell$.  This proves every assertion.
\end{proof}

\begin{theorem}[Medium-Prime Split Tower]\label{thm:medium-tower}
There exist absolute constants
\[
       T\in\Z_{\ge 16},
       \qquad H\ge1,
       \qquad C_{\rm ar}\in\Z_{\ge2},
\]
with
\[
       \gamma:=T\log 2-\log H> T H_2(4/T),
       \qquad
       H_2(\rho):=-\rho\log\rho-(1-\rho)\log(1-\rho),
\]
such that the following holds.
(We use $H_2$ to denote binary entropy with natural logarithms to distinguish from the constant $H$ and class number $h$).
For every sufficiently large parameter $Y$ there are $T$ distinct rational primes
\[
       q_1,\ldots,q_T\in[Y,2Y]
\]
and an infinite sequence of CM fields $K_s=L_s(i)$, where $L_s$ is totally real, satisfying:
\begin{enumerate}[label=(\roman*)]
    \item each $q_t$ splits completely in every $K_s$; in particular $q_t\equiv1\pmod4$ and every prime of $K_s$ above $q_t$ has residue field $\F_{q_t}$;
    \item if $f_s:=[L_s:\Q]$, then $f_s\to\infty$ and
    \[
          h(K_s)\le H^{f_s};
    \]
    \item for every sufficiently large $b$ there is a layer $s=s(b,Y)$ with
    \[
          C_{\rm ar}b\le f_s\le 3C_{\rm ar}b .
    \]
\end{enumerate}
The constants $T,H,C_{\rm ar}$ are independent of $Y$ and $b$.
\end{theorem}

\begin{proof}
We derive the statement from \Cref{prop:ud-package}.  For $\rho\le1/2$,
\[
       H_2(\rho)=O\!\left(\rho\log\frac1\rho\right),
\]
and therefore, for $T\to\infty$,
\[
       T H_2(4/T)=O(\log T).
\]
For the tower parameter $\ell$, put
\[
       T_\ell^*:=\floor{(c_0/2)\ell^2}.
\]
Then $T_\ell^*\le T_\ell$ (recall $T_{\ell} := \lfloor (\ell - 1)^2/100 \rfloor$ and $c_0 = 1/200$) and, along the tower parameters of \Cref{prop:ud-package},
\[
       T_\ell^*\log 2-\log H_\ell
       \ge \bigl((c_0/2)\ell^2-1\bigr)\log2-C_H\ell\log\ell
       =\Omega(\ell^2).
\]
Choose once and for all a sufficiently large $\ell$ such that $T_\ell^*\ge16$ and
\[
       T_\ell^*\log2-\log H_\ell>T_\ell^* H_2 (4/T_\ell^*).
\]
Set
\[
       T:=T_\ell^*,\qquad H:=H_\ell,\qquad N:=N_\ell,\qquad d_0:=d_\ell = 3,
       \qquad C_{\rm ar}:=\max\{2,d_0\}.
\]
These constants now depend only on the objects fixed in \Cref{prop:ud-package}.

Let $Y$ be sufficiently large.  
We now use only the standard Chebotarev Density Theorem for the fixed finite normal extension $N/\Q$: if $\pi_N(x)$ counts rational primes $q\le x$ that split completely in $N$, then
\begin{equation}
    \label{eq:prime-density}
    \pi_N(x)\sim \frac{1}{[N:\Q]}\frac{x}{\log x} \text{.}
\end{equation}
   
To see this, recall the Chebotarev Density Theorem below.

\begin{definition}[Density (see e.g. \cite{milne2020algebraic})]
    \label{def:density}
    Let $S$ be a set of rational primes, and let $P$ be the set of all rational primes. 
    We say that $S$ has density $\delta$ if
    \begin{equation*}
        \delta = \lim_{n \rightarrow \infty} \frac{|\set{p \in S \mid p \leq n}|}{|\set{p \in P \mid p \leq n}|} \text{.}
    \end{equation*}
\end{definition}

\begin{theorem}[Chebotarev Density Theorem (see e.g. \cite{milne2020algebraic})]
    \label{thm:chebotarev-density}
    Let $N/\Q$ be a finite Galois extension.
    The primes that split completely in $N$ have density $1/[N:\Q]$.
\end{theorem}

In particular, combined with the Prime Number Theorem, we have \eqref{eq:prime-density}.
Subtracting the asymptotics at $2Y$ and $Y$ gives
\[
       \#\{q\in[Y,2Y]:q\text{ splits completely in }N\}
       \sim \frac{1}{[N:\Q]}\frac{Y}{\log Y}.
\]
For all large $Y$ this number is at least the fixed integer $T$.  Choose distinct primes
\[
       q_1,\ldots,q_T\in[Y,2Y]
\]
splitting completely in $N$, and put $\mathcal Q:=\{q_1,\ldots,q_T\}$. 
Applying \Cref{prop:ud-package} to this $\mathcal Q$ gives an infinite pro-$3$ Galois extension over $L_{0,\ell}$.  
Because \Cref{prop:ud-package} is uniform in $\mathcal Q$, the same $H$ works for all choices of $Y$ and all choices of the dyadic split primes.

Every $q_t$ splits completely in each totally real finite subextension $L/L_{0, \ell}$ by \Cref{prop:ud-package}.  
Also $N \supseteq \Q(i)$, so $q_t$ splits in $\Q(i)$ since $\Q(i)$ is an intermediate extension.
Hence $q_t \equiv 1 \pmod{4}$ by Euler's criterion.
We claim that $q_t$ splits completely in $K=L(i)$, so that every prime of $K$ above $q_t$ has residue field $\F_{q_t}$. 
Indeed $q_t$ splits completely in $L$ and $q_t \equiv 1 \pmod{4}$ so consider a prime $\frakp$ above $q_t$ in $L$.
$q_t \equiv 1 \pmod{4}$ implies that $x^2 + 1$ splits into linear factors in $\calO_{L}/\frakp \cong \F_{q_t}$.
Then, each $\frakp$ splits completely into two distinct primes in $K$ (see e.g. Theorem 27 of \cite{marcus1977number}), implying the complete splitting of $q_t$ in $K$.

The same \Cref{prop:ud-package} gives
\[
       h(K)\le H^{[L:\Q]}
\]
for every such layer.

It remains to regularize the degrees.  
Let $G := \Gal(\mathcal L_{\ell,\mathcal Q}/L_{0,\ell})$, an infinite pro-$3$ group.  
For every $i \ge 0$, infinitude of $G$ gives a finite $3$-group quotient $P$ of order at least $3^i$. 
Since every finite $3$-group has a normal subgroup of index $3^i$ whenever its order is at least $3^i$ (see e.g. Theorem 6.1 of \cite{dummit2004abstract}),
pulling such a subgroup back to $G$ gives a normal subgroup $G_i \triangleleft G$ that is both open and closed with
\[
       [G:G_i]=3^i \text{.}
\]
Let $L_i$ be the corresponding finite layer (i.e. the fixed field of $G_i$) promised by the infinite Galois correspondence (Theorem 7.13 of \cite{milne2003fields}).  
Then, 
\[
       [L_i:\Q] = [L_i : L_{0, \ell}] [L_{0, \ell}: \Q] = d_0\,3^i \text{.}
\]
For every sufficiently large $b$, choose the first $i$ with
\[
       d_0\,3^i\ge C_{\rm ar}b .
\]
Since $b$ is large, $i\ge1$, and minimality gives
\[
       d_0\,3^{i-1}<C_{\rm ar}b.
\]
Thus
\[
       C_{\rm ar}b\le [L_i:\Q]=d_0\,3^i<3C_{\rm ar}b.
\]
Taking the sequence of all such finite layers proves the theorem, with constants independent of $Y$ and $b$.
\end{proof}

\begin{theorem}[Medium-Prime Norm-One $S$-Unit Fiber]\label{thm:medium-input}
There exist absolute constants
\[
        T\in\Z_{\ge 16},
        \qquad C_{\rm ar}\in\Z_{\ge2},
        \qquad \gamma>0,
\]
with
\[
        \gamma>T H_2(4/T),
        \qquad
        H_2(\rho):=-\rho\log \rho-(1-\rho)\log(1-\rho),
\]
such that the following holds for every sufficiently large parameter $Y$ and every sufficiently large integer $b$.
There is a CM field $K$ with complex conjugation $c$, an integer $f$ with
\[
        C_{\rm ar}b\le f\le 3C_{\rm ar}b,
        \qquad [K:\Q]=2f,
\]
and $T$ distinct rational primes $q_1,\ldots,q_T\in [Y,2Y]$ satisfying:
\begin{enumerate}[label=(\roman*)]
    \item each $q_t$ splits completely in $K$; in particular every prime of $K$ above $q_t$ has residue field $\F_{q_t}$, and $q_t\equiv1\pmod4$;
    \item the $T$ rational primes determine $Tf$ conjugate prime-ideal pairs
    \[
        \{\mathfrak P_m,c\mathfrak P_m\},\qquad m=1,\ldots,Tf,
    \]
    in $K$;
    \item there is a vector $\eta\in\{0,1\}^{Tf}$, a family $\mathcal F\subseteq\{0,1\}^{Tf}$ with
    \[
        |\mathcal F|\ge \exp(\gamma f),
    \]
    and, for every $\eps\in\mathcal F$, an element $u_\eps\in K^\times$ such that, with $Q:=\prod_{t=1}^Tq_t$,
    \begin{align*}
        u_\eps c(u_\eps)&=1,\qquad
        |\sigma(u_\eps)|=1\quad\text{for every complex embedding }\sigma,\\
        Q^2u_\eps&\in\OO_K,\\
        \val_{\mathfrak P_m}(u_\eps)&=2(\eps_m-\eta_m),\qquad
        \val_{c\mathfrak P_m}(u_\eps)=-2(\eps_m-\eta_m)
    \end{align*}
    for every $m=1,\ldots,Tf$.
\end{enumerate}
All constants are independent of $Y$ and $b$.
\end{theorem}

\begin{proof}
Apply \Cref{thm:medium-tower} with the given $Y$, and choose the layer with $f=[L_s:\Q]$ in the interval $[C_{\rm ar}b,3C_{\rm ar}b]$.  Put $K=K_s=L_s(i)$.  
Each $q_t$ splits completely in $K$, hence gives $2f$ primes of residue degree and ramification index one.
We claim that complex conjugation permutes these primes and fixes none of them.
Recall that complex conjugation preserves prime ideals (\Cref{def:conjugate-prime}) and complex conjugation fixes rational primes, so that if $\frakP$ lies above $(q_t)$, so does $c \frakP$.

We claim that \(c\) fixes no prime \(\frakP \mid q_t\). 
Suppose \(c\frakP=\frakP\). 
Since \(q_t\) splits completely in \(K\), we have \(e(\frakP|q_t)=f(\frakP|q_t)=1\). 
Hence, by the
local degree formula for completions (\Cref{clm:local-degree-formula}),
\[
[K_{\frakP} : \mathbb Q_{q_t}] = e(\frakP|q_t) f(\frakP|q_t)=1,
\]
so \(K_{\frakP} = \mathbb Q_{q_t}\). 
Consider the injective map $c: K \hookrightarrow K_{c\frakP} = K_{\frakP}$ given by embedding $K$ into $K_{\frakP}$.
Since \(c\frakP=\frakP\), the automorphism \(c\) preserves the
\(\frakP\)-adic valuation and therefore extends to a map $K_{\frakP} \rightarrow K_{\frakP}$ (\Cref{clm:completion-homomorphism}) that fixes $\Q_{q_t} \subseteq K_{\frakP}$. 
But \(K_{\frakP} = \mathbb Q_{q_t}\), so this extension is the identity. 
Hence \(c\) is the identity on \(K\), contradicting that \(c\) is the complex conjugation automorphism.
Thus \(c\frakP \ne \frakP\).

Since $c^2 = 1$, the $2f$ primes form exactly $f$ conjugate pairs.  
Enumerate these pairs as
\[
       \{\mathfrak P_m,c\mathfrak P_m\},\qquad m=1,\ldots,Tf .
\]
The class-number bound gives $h(K)\le H^f$.  Let
\[
       \gamma:=T\log 2-\log H.
\]
By \Cref{thm:medium-tower}, $\gamma>T H_2(4/T)>0$.

For each $\eps\in\{0,1\}^{Tf}$ define an integral ideal
\[
       \mathfrak A_\eps:=\prod_{m=1}^{Tf}\mathfrak P_m^{\eps_m}(c\mathfrak P_m)^{1-\eps_m} \text{.}
\]
There are $2^{Tf}$ such ideals and at most $h(K)\le H^f$ ideal classes.  
Hence some ideal class contains a fiber
\[
       \mathcal F\subseteq\{0,1\}^{Tf},
       \qquad
       |\mathcal F|\ge \frac{2^{Tf}}{H^f}=\exp(\gamma f).
\]
Fix one element $\eta\in\mathcal F$.  For every $\eps\in\mathcal F$, the fractional ideal
\[
       \mathfrak A_\eps\mathfrak A_\eta^{-1}
\]
is principal.  Choose $v_\eps\in K^\times$ with
\[
       (v_\eps)=\mathfrak A_\eps\mathfrak A_\eta^{-1},
\]
and set
\[
       u_\eps:=v_\eps/c(v_\eps).
\]
Then $u_\eps c(u_\eps)=1$. For any complex embedding $\sigma$, the embeddings $\sigma$ and $\sigma\circ c$ are complex conjugates on the CM field, so
\[
       |\sigma(u_\eps)|=\frac{|\sigma(v_\eps)|}{|\sigma(c(v_\eps))|}
       =\frac{|\sigma(v_\eps)|}{|\overline{\sigma(v_\eps)}|}=1.
\]
 Passing from $v_\eps$ to $v_\eps/c(v_\eps)$ doubles these valuations, giving the displayed valuation pattern for $u_\eps$.

We now prove that \(Q^2u_\varepsilon\in\mathcal O_K\). 
It suffices to show that \(\operatorname{val}_{\mathfrak P}(Q^2u_\varepsilon)\ge 0\) for every finite prime ideal \(\mathfrak P\subseteq\mathcal O_K\).
Recall that \((v_\varepsilon)=\frakA_\varepsilon \frakA_\eta^{-1}\), where \(\frakA_\varepsilon=\prod_{m=1}^{Tf}\frakP_m^{\varepsilon_m}(c\frakP_m)^{1-\varepsilon_m}\) and \(\frakA_\eta=\prod_{m=1}^{Tf}\frakP_m^{\eta_m}(c\frakP_m)^{1-\eta_m}\). 
Thus the only prime ideals at which \(v_\varepsilon\) can have nonzero valuation are the primes \(\frakP_m\) and \(c\frakP_m\), all of which lie above the selected rational primes \(q_1,\dots,q_T\). 
Moreover, for each \(m\), we have \(\operatorname{val}_{\frakP_m}(v_\varepsilon)=\varepsilon_m-\eta_m\) and \(\operatorname{val}_{c\frakP_m}(v_\varepsilon)=(1-\varepsilon_m)-(1-\eta_m)=-(\varepsilon_m-\eta_m)\). 
Since \(\varepsilon_m,\eta_m\in\{0,1\}\), the valuations all lie in \(\{-1,0,1\}\).
Furthermore, $\val_{\frakP_{m}}(u_{\eps}) = \val_{\frakP_{m}}(v_{\eps}/c(v_{\eps})) = 2 (\eps_m - \eta_m)$ and similarly for $\val_{c \frakP_{m}}(u_{\eps}) = - 2 (\eps_m - \eta_m)$.

Now let \(\mathfrak P\) be any finite prime of \(K\). If \(\mathfrak P\) does not lie above any of the selected rational primes \(q_t\), then \(\operatorname{val}_{\mathfrak P}(v_\varepsilon)=0\) and, since complex conjugation permutes the selected primes, \(\operatorname{val}_{\mathfrak P}(c(v_\varepsilon))=0\). Hence \(\operatorname{val}_{\mathfrak P}(u_\varepsilon)=\operatorname{val}_{\mathfrak P}(v_\varepsilon/c(v_\varepsilon))=0\). Since \(Q^2\in\mathcal O_K\), we get \(\operatorname{val}_{\mathfrak P}(Q^2u_\varepsilon)\ge 0\).

It remains to consider the case where \(\mathfrak P\mid q_t\) for some selected rational prime \(q_t\). At such a prime, both \(\operatorname{val}_{\mathfrak P}(v_\varepsilon)\) and \(\operatorname{val}_{\mathfrak P}(c(v_\varepsilon))\) lie in \(\{-1,0,1\}\). Therefore \(\operatorname{val}_{\mathfrak P}(u_\varepsilon)=\operatorname{val}_{\mathfrak P}(v_\varepsilon)-\operatorname{val}_{\mathfrak P}(c(v_\varepsilon))\ge -2\). On the other hand, \(q_t\) splits completely in \(K\). In particular, \(q_t\) is unramified in \(K\), so every prime \(\mathfrak P\mid q_t\) has ramification index \(1\). Hence \(\operatorname{val}_{\mathfrak P}(q_t)=1\). Since \(Q=\prod_{t=1}^T q_t\), we have \(\operatorname{val}_{\mathfrak P}(Q^2)=2\). Thus \(\operatorname{val}_{\mathfrak P}(Q^2u_\varepsilon)=\operatorname{val}_{\mathfrak P}(Q^2)+\operatorname{val}_{\mathfrak P}(u_\varepsilon)\ge 2-2=0\).
Therefore \(\operatorname{val}_{\mathfrak P}(Q^2u_\varepsilon)\ge 0\) for every finite prime \(\mathfrak P\subseteq\mathcal O_K\). Hence \(Q^2u_\varepsilon\) is integral, and so \(Q^2u_\varepsilon\in\mathcal O_K\).

\end{proof}

\NumberTheoryProp*

\begin{proof}
Apply \Cref{thm:medium-input} with $Y=\ceil{4\sqrt L}$ and the given $b$.  We get a CM field $K$, a degree parameter $f\in[C_{\rm ar}b,3C_{\rm ar}b]$, $T$ rational primes $q_t\in[Y,2Y]$ splitting completely in $K$, and a family
\[
        \mathcal F\subseteq\{0,1\}^{Tf},\qquad |\mathcal F|\ge e^{\gamma f},
\]
with associated elements $\set{u_\eps}_{\eps \in \calF}$ and $\eta$.

We first extract many shattered valuation coordinates.  
Let
\[
        M:=Tf,
        \qquad s:=(3C_{\rm ar}+2)b.
\]
Since $f\ge C_{\rm ar}b$, we have
\[
        \frac{s}{M}
        \le \frac{(3C_{\rm ar}+2)b}{TC_{\rm ar}b}
        \le \frac4T
        \le \frac12.
\]
The Sauer--Shelah lemma says that if $\mathcal F$ shatters no set of size $s$, then
\[
        |\mathcal F|\le \sum_{i=0}^{s-1}\binom{M}{i}
        \le \exp\bigl(M H_2(s/M)\bigr)
        \le \exp\bigl(Tf H_2(4/T)\bigr) \text{.}
\]
The entropy slack $\gamma > T H_2(4/T)$ in \Cref{thm:medium-input} makes this strictly smaller than $e^{\gamma f}\le |\mathcal F|$.  Hence $\mathcal F$ shatters a coordinate set $S$ of size at least $(3C_{\rm ar}+2)b$.

Each rational-prime label contributes exactly $f\le3C_{\rm ar}b$ coordinates.  
Our goal is to find a subset of $S$ of size at least $2b$ that contains at most $b$ coordinates of each rational-prime label.
If at least two labels have more than $b$ selected coordinates in $S$, then after truncation those two labels alone contribute $2b$ retained coordinates. 
If at most one label has more than $b$ selected coordinates in $S$, then the total number lost is at most $f-b\le(3C_{\rm ar}-1)b$, so the retained number is at least
\[
       (3C_{\rm ar}+2)b-(3C_{\rm ar}-1)b=3b.
\]
Thus in all cases at least $2b$ coordinates remain, and no label occurs more than $b$ times among them.  Retain exactly $2b$ of the surviving coordinates.  This retained set is a subset of the shattered set $S$, and hence is itself shattered.  

We claim that these coordinates can be partitioned into $b$ pairs whose two members have distinct rational-prime labels.  Indeed, suppose an even number $2m$ of coordinates remain and every label class has size at most $m$.  If the maximum multiplicity is exactly $m$, then at most two label classes have this multiplicity.  Remove one coordinate from a largest label class and one coordinate from a different largest class if there are two, or from any different label class if there is only one; in the one-largest-class case such a different class exists whenever $m>0$, because only $m$ of the $2m$ coordinates have the largest label.  All remaining label classes then have size at most $m-1$.  If the maximum multiplicity is already less than $m$, then it is at most $m-1$; remove any two coordinates with distinct labels, which exist unless $m=0$.  Again every remaining label class has size at most $m-1$, half of the remaining $2m-2$ coordinates.  Induction on $m$ (beginning with $m = b$) proves the claim.

Label the resulting pairs by $j=1,\ldots,b$, and write the two coordinates in the $j$th pair as $m(j,1)$ and $m(j,2)$.  For each $r\in\{1,2\}$, define the selected prime ideal
\[
      \mathfrak p_{j,r}:=
      \begin{cases}
      \mathfrak P_{m(j,r)} & \text{if }\eta_{m(j,r)}=0,\\
      c\mathfrak P_{m(j,r)} & \text{if }\eta_{m(j,r)}=1.
      \end{cases}
\]
Now fix $a\in\{0,1\}^b$.  For each selected coordinate $m(j,r)$ prescribe the bit
\[
      \theta_{j,r}(a):=
      \begin{cases}
      a_j & \text{if }\eta_{m(j,r)}=0,\\
      1-a_j & \text{if }\eta_{m(j,r)}=1.
      \end{cases}
\]
Because the retained coordinates form a shattered subset of $S$, there exists $\eps(a)\in\mathcal F$ with
\[
      \eps(a)_{m(j,r)}=\theta_{j,r}(a)\qquad\text{for all }j\in[b],\ r\in\{1,2\}.
\]
Define $u_a:=u_{\eps(a)}$.  If $\eta_{m(j,r)}=0$, then
\[
      \val_{\mathfrak p_{j,r}}(u_a)=\val_{\mathfrak P_{m(j,r)}}(u_{\eps(a)})
      =2(\eps(a)_{m(j,r)}-\eta_{m(j,r)})=2a_j.
\]
If $\eta_{m(j,r)}=1$, then
\[
      \val_{\mathfrak p_{j,r}}(u_a)=\val_{c\mathfrak P_{m(j,r)}}(u_{\eps(a)})
      =-2(\eps(a)_{m(j,r)}-\eta_{m(j,r)})=2a_j.
\]
In both cases the conjugate prime has valuation $-2a_j$, so \eqref{eq:val1}--\eqref{eq:val2} hold.  
The two rational-prime labels assigned to bit $j$ are distinct and both lie in $[Y,2Y]$, so their product is at least $Y^2>L$.  Finally, $[K:\Q]=2f\le6C_{\rm ar}b$, so the degree bound follows after setting $C_0:=6C_{\rm ar}$.
\end{proof}

\section{Reductions from \texorpdfstring{$\Z$}{Z}-OV}
\label{app:reductions}

In this appendix, we give reductions from $\Z$-OV to $\Z$-Max-IP, Furthest Pair, and Bichromatic Closest Pair for completeness. 

\begin{lemma}[Implicit in \cite{DBLP:conf/soda/Williams18}, Theorem 4.3 in \cite{DBLP:journals/toc/Chen20}]\label{lem:zov-to-maxip}
There is an $O(n\poly(r))$-time reduction from $\Z$-OV in dimension $r$ to exact $\Z$-Max-IP in dimension $O(r^2)$.  If the $\Z$-OV coordinates have bit-length $B$, then the produced $\Z$-Max-IP coordinates have bit-length $O(B)$.
\end{lemma}

\begin{proof}
    Represent
    \[
          P(x,y)=-(x\cdot y)^2=-\sum_{i,j}x_ix_jy_iy_j
    \]
    as an ordinary inner product in dimension $r^2$, for example by mapping $x$ to coordinates $x_ix_j$ and $y$ to coordinates $-y_iy_j$.  Then $P(x,y)\le0$ for all pairs, and $P(x,y)=0$ exactly when $x\cdot y=0$.  Thus a zero inner product exists in the $\Z$-OV instance iff the maximum produced inner product is $0$.  Products of two $B$-bit integers have $O(B)$ bits, so the bit-length stays $O(B)$.
\end{proof}

Next, we give a reduction from $\Z$-OV to Furthest Pair, which is based on \cite[Theorem 1.2]{DBLP:conf/soda/Williams18}. 

\begin{lemma}
\label{lem:zov-to-furthest-pair}
    There is an $O(n \poly(r))$-time reduction from $\Z$-OV in dimension $r$ to Furthest Pair in dimension $O(r^2)$.
    If the $\Z$-OV coordinates have bit-length $B$, then the produced Furthest Pair coordinates have bit-length $O(B + \log r)$.
\end{lemma}
\begin{proof}
    We first run the reduction in \cref{lem:zov-to-maxip} to obtain a $\Z$-Max-IP instance in dimension $O(r^2)$ between sets of vectors $U$ and $V$, where  $u \cdot v \le 0$ for $u \in U$ and $v \in V$, and the goal is to determine if the maximum dot product is $0$. If there is any zero vector, the maximum dot product is trivially $0$. So in the following, we assume all vectors are nonzero. 
    
    Let $C$ be a sufficiently large integer. 
    For any $u \in U$ and $v \in V$, define 
    \[
    u' = (u / \lVert u \rVert_2, C, 0), v' = (-v / \lVert v \rVert_2, 0, C). 
    \]
    Let the set of $u'$ and $v'$ be the Furthest Pair instance.  

    For $u_1'$ and $u_2'$, their squared distance is 
    \[
    \lVert u_1' - u_2'\rVert_2^2 = \lVert u_1 / \lVert u_1 \rVert_2 - u_2 / \lVert u_2 \rVert_2 \rVert_2^2 \le 4, 
    \]
    and similarly, 
    \[
    \lVert v_1' - v_2'\rVert_2^2 \le 4.  
    \]
    Next consider 
    \[
    \lVert u' - v'\rVert_2^2 = \lVert u'\rVert_2^2 + \lVert v'\rVert_2^2 - 2 u' \cdot v' = 2(C^2 + 1) + \frac{2 u \cdot v}{\lVert u \rVert_2 \lVert v \rVert_2}. 
    \]
    For large enough $C$, the furthest pair must be $u', v'$ for some $u \in U$ and $v \in V$. Furthermore, if the furthest pair has squared distance $2(C^2 + 1)$, then there exists $u \in U, v \in V$ with $u \cdot v = 0$; otherwise, there does not exist such pair. 

    The above completes the reduction to Furthest Pair with real coordinates. To convert the vectors to integer coordinates, first consider if $u \cdot v < 0$, then 
    \[
    \frac{u \cdot v}{\lVert u \rVert_2 \lVert v \rVert_2} \le -\frac{1}{\poly(r) \cdot 2^{O(B)}}. 
    \]
    Let $\Delta$ be the absolute value of the above gap.  
    Next, choose $\delta > 0$ be some small real number to be determined. If we round each coordinate $z$ to a nearest multiple $k(z) \delta$ of  $\delta$ for $k(z) \in \Z$, then the squared distance between any two vectors can change by at most $O(\delta r^2)$. Hence, we can have $\delta = O(\Delta/r^2) = \frac{1}{\poly(r) \cdot 2^{O(B)}}$, so that the squared distance between any two vectors changes by at most $\Delta / 2$. This way, we can determine whether the maximum inner product before rounding is $0$, by checking whether the maximum inner product after rounding is larger than $-\Delta / 2$. Finally, after rounding $z$ to $k(z) \delta$, output the integer coordinate $k(z)$ to obtain integer vectors. 
\end{proof}

Finally, we give a reduction from $\Z$-OV to Bichromatic Closest Pair, which is based on \cite[Corollary 2.1]{DBLP:conf/soda/Williams18}.

\begin{lemma}
    There is an $O(n \poly(r))$-time reduction from $\Z$-OV in dimension $r$ to Bichromatic Closest Pair in dimension $O(r^2)$.
    If the $\Z$-OV coordinates have bit-length $B$, then the produced Bichromatic Closest Pair coordinates have bit-length $O(B + \log r)$.
\end{lemma}
\begin{proof}
    We first run the reduction in \cref{lem:zov-to-maxip} to obtain a $\Z$-Max-IP instance in dimension $O(r^2)$ between sets of vectors $U$ and $V$, where $u \cdot v \le 0$ for $u \in U$ and $v \in V$, and the goal is to determine if the maximum dot product is $0$. If there is any zero vector, the maximum dot product is trivially $0$. So in the following, we assume all vectors are nonzero.  
    
    For any $u \in U$ and $v \in V$, define 
    \[
    u' = u / \lVert u \rVert_2, v' = v / \lVert v \rVert_2. 
    \]
    Let $U' = \{u': u \in U\}$ and $V' = \{v': v \in V\}$ be the input to Bichromatic Closest Pair. 

    Consider 
    \[
    \lVert u' - v'\rVert_2^2 = \lVert u'\rVert_2^2 + \lVert v'\rVert_2^2 - 2 u' \cdot v' = 2 - \frac{2 u \cdot v}{\lVert u \rVert_2 \lVert v \rVert_2}. 
    \]
    If the closest pair has squared distance $2$, then there exists $u \in U, v \in V$ with $u \cdot v = 0$; otherwise, there does not exist such pair. 

    The above completes the reduction to Bichromatic Closest Pair with real coordinates. The vectors can be rounded to integers in the same way as \cref{lem:zov-to-furthest-pair}. 
\end{proof}

\section{Details on Lean Formalization of \texorpdfstring{\cref{lem:nonuiform-local-lean}}{}}
\label{app:lean}

\begin{figure}[htbp]
    \centering

    \begin{subfigure}{0.48\textwidth}
        \centering
        \includegraphics[width=\linewidth]{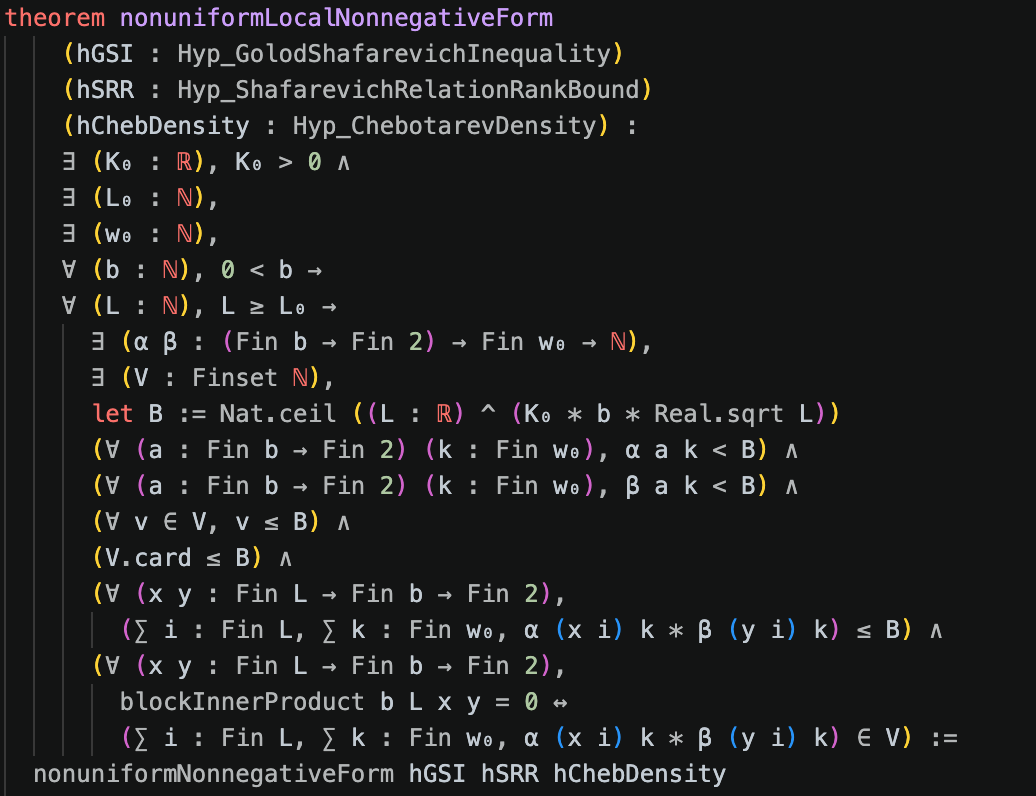}
        \caption{}
        \label{fig:a}
    \end{subfigure}
    \hfill
    \begin{subfigure}{0.48\textwidth}
        \centering
        \includegraphics[width=\linewidth]{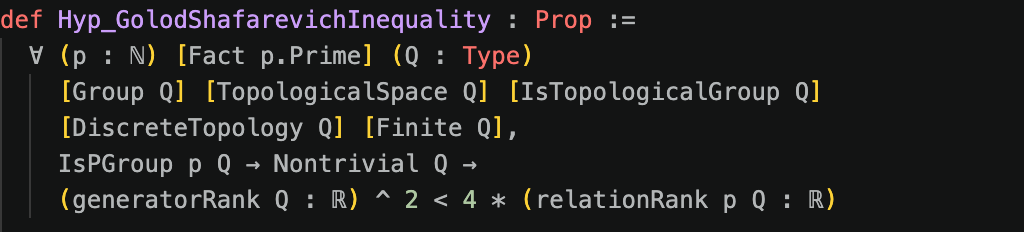}
        \caption{}
        \label{fig:b}
    \end{subfigure}

    \vspace{0.5em}

    \begin{subfigure}{0.48\textwidth}
        \centering
        \includegraphics[width=\linewidth]{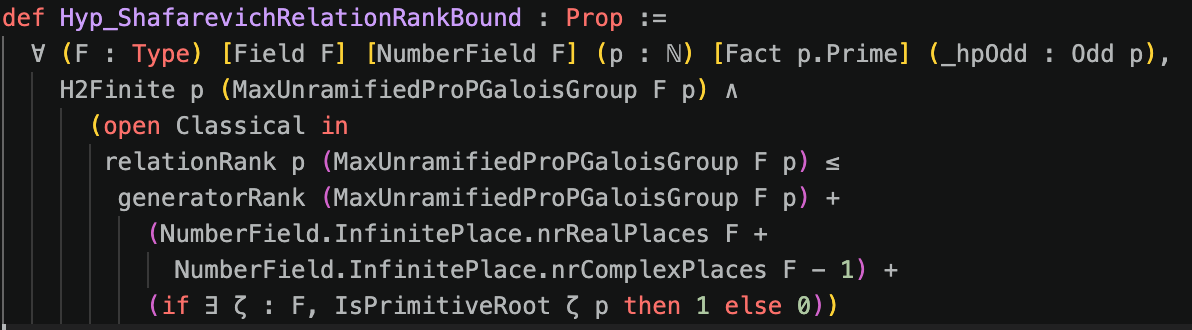}
        \caption{}
        \label{fig:c}
    \end{subfigure}
    \hfill
    \begin{subfigure}{0.48\textwidth}
        \centering
        \includegraphics[width=\linewidth]{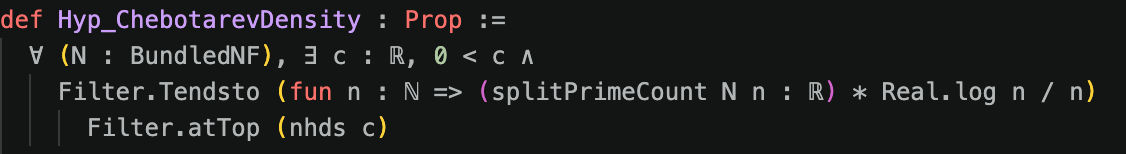}
        \caption{}
        \label{fig:d}
    \end{subfigure}

    \caption{Some code snippets of the Lean 4 formalization.}
    \label{fig:four_subfigures}
\end{figure}

We used Aristotle \cite{achim2025aristotleimolevelautomatedtheorem} to formalize the proof of \cref{lem:nonuiform-local-lean} in Lean 4, which depends on Aleph Prover's formalization \cite{erdos_unit_distance_formalization_2026} of \cite{UnitDistance}. 

The main result formalized by Lean 4 is shown in \cref{fig:a}, which proves \cref{lem:nonuiform-local-lean} assuming \texttt{Hyp\_GolodShafarevichInequality}, \texttt{Hyp\_ShafarevichRelationRankBound} and \texttt{Hyp\_ChebotarevDensity}, the first two of which are also assumed in \cite{erdos_unit_distance_formalization_2026}. These are standard results in algebraic number theory:
\begin{enumerate}
    \item Golod-Shafarevich inequality (\cref{fig:b}): For every prime $p$, and every nontrivial finite $p$-group $Q$, $d(Q)^2 < 4 r(Q)$, where $d(Q)$ and $r(Q)$ denote the generator rank and relation rank of $Q$ respectively. This is exactly \cref{prop:golod-shafarevich}. See, e.g. \cite{golod1964class,GolodShafarevich1965}, for references. 
    \item Shafarevich's relation-rank bound (\cref{fig:c}): for every number field $F$ and every odd prime $p$, let $G=
\Gal(F^{\ur,p} / F)$, then the following holds: 
    \begin{itemize}
        \item The second continuous cohomology group $H^2(G, \F_p)$ is finite dimensional over $\F_p$. 
        \item $r(G) \le d(G) + r_1(F) + r_2(F) - 1 + \delta_p(F)$, where $r_1(F)$ denotes the number of real places of $F$, $r_2(F)$ denotes the number of complex places of $F$, and $\delta_p(F)$ is $1$ iff $F$ contains a $p$-th primitive root of unity and $0$ otherwise. 
    \end{itemize}
    \cref{prop:shafarevich} is a special case of it where $F$ is a totally real cubic field and $p = 3$. See e.g., \cite{vsafarevivc1963extensions, shafarevich1966extensions, mayer2015new} for references. 
    \item Chebotarev density theorem (\cref{fig:d}): For every number field $N$, there exists a positive constant $c$ such that the following holds: let $\pi_N(n)$ be the number of primes up to $n$ that split completely in $N$, then 
    $\frac{\pi_N(n)}{n} \cdot \log n \sim c$. This can be deduced from \cref{thm:chebotarev-density}: Let $\widetilde{N}$ be the Galois closure of $N / \Q$, then \cref{thm:chebotarev-density} implies $c = 1 / [\widetilde{N} : \Q]$. 
    See e.g., \cite{tschebotareff1926bestimmung,neukirch2013algebraic,milne2020algebraic}, for references. 
\end{enumerate}

\end{document}